\documentclass[journal,twoside,web]{ieeecolor}
\usepackage{generic}
\usepackage{cite}
\usepackage{amsmath,amssymb,amsfonts}
\usepackage{graphicx}
\usepackage{textcomp}
\usepackage{color,array}

\usepackage{blindtext}
\usepackage{mathrsfs}
\usepackage{amssymb,setspace}
\usepackage{amsmath,bm}
\usepackage{amsfonts}
\usepackage{epsfig}
\usepackage{enumerate}
\usepackage{graphicx}
\usepackage{subfigure}
\usepackage{arydshln}
\usepackage{lipsum}
\usepackage{float}
\usepackage{galois}
\usepackage{color, soul}
\usepackage{float}
\usepackage{algorithm}  
\usepackage{algorithmicx}  
\usepackage[noend]{algpseudocode}
\usepackage{romannum}
\usepackage{url}
\usepackage{hyperref}
\newtheorem{Thm}{Theorem}
\newtheorem{Asu}{Assumption}

\newtheorem{Def}{Definition}
\newtheorem{Lem}{Lemma}

\newtheorem{Rem}{Remark}

\setulcolor{red}
\setstcolor{green}
\sethlcolor{blue}

\hypersetup{hidelinks=true}
\usepackage{textcomp}
\def\BibTeX{{\rm B\kern-.05em{\sc i\kern-.025em b}\kern-.08em
    T\kern-.1667em\lower.7ex\hbox{E}\kern-.125emX}}
\begin{document}
\title{Robust control synthesis for uncertain linear systems with input saturation using mixed IQCs}
\author{Xu Zhang and Fen Wu
\thanks{Xu Zhang is with the School of
Electrical Engineering and Computer Science, The Pennsylvania State
	University, University Park, PA 16802. }% <-this % stops a space
\thanks{Fen Wu (Corresponding author) is with Department of Mechanical and Aerospace Engineering, North Carolina State University, Raleigh, NC, 27606 USA e-mail: fwu@ncsu.edu.}}

\maketitle

\begin{abstract}
This paper develops a robust control synthesis method for uncertain linear systems with input saturation in the framework of integral quadratic constraints (IQCs). The system is reformulated as a linear fractional representation (LFR) that captures both dead-zone nonlinearity and time-varying uncertainties. By combining mixed IQC-based dissipation inequalities with quadratic Lyapunov functions, sufficient conditions for robust stabilization are established. Compared with conventional approaches based on a single static sector condition for the dead-zone nonlinearity, the proposed method yields improved $\mathcal{L}_2$-gain performance through the use of scaled mixed IQCs. For systems subject to time-varying structured uncertainties, a new scaled bounded real lemma is further developed based on the IQC characterization. The resulting $\mathcal{H}_\infty$ synthesis conditions are expressed as linear matrix inequalities (LMIs), which are numerically tractable in all decision variables, including the scaling factors in the IQC multipliers. The proposed method is validated using a second-order uncertain system in linear fractional form, and its superiority over an anti-windup design is further illustrated by a cart-pendulum example.
\end{abstract}


\section{Introduction}
\IEEEPARstart{C}{yber-physical} systems tightly integrate physical processes with sensing, communication, computation, and control, and have been widely deployed in a broad range of safety-critical and high-performance applications, such as robotics \cite{Siciliano2016}, unmanned aerial vehicles (UAVs) \cite{Zhao2018}, intelligent transportation systems \cite{Oladimeji2023}, and advanced manufacturing platforms including lithography equipment \cite{Zhang2012}. In these systems, the interaction between cyber components and physical dynamics is typically strong and bidirectional, so that uncertainty, disturbances, and physical limitations can have a direct impact on closed-loop behavior. As a consequence, the analysis and control of CPSs need explicitly account for both dynamical complexity and implementation constraints in order to ensure stability, performance, and reliability in practical operation.

Among the various nonideal factors arising in CPSs, input saturation is one of the most common and practically important nonlinear input constraints. It originates from the inherent physical limitations of actuators and can severely degrade system performance, enlarge transient responses, or even lead to instability; see, e.g., \cite{KA1994,PTG2010,TZ2001,DTA2009} and the references therein. For this reason, input saturation has attracted sustained attention in control engineering and related disciplines over the past decades. In particular, when input saturation coexists with uncertainty and external disturbances in linear systems, the feedback controller is expected not only to guarantee robust closed-loop stability, but also to provide satisfactory disturbance attenuation, which is commonly quantified in terms of the $\mathcal{L}_2$ gain \cite{HZT2004}.

\textit{Literature review.} Over the past decades, a variety of control design methods have been developed to address saturation problems in linear and nonlinear systems, as well as in stable and unstable plants \cite{YZT2002,Z1998,Fenwu2009}. First, anti-windup design has been widely recognized as an effective approach for compensating input saturation; see \cite{KTD1998, T1999, WL2004, GJIA2003, TAL2008} and the references therein. For example, \cite{LL2014} proposed a saturation-based switching anti-windup design to enlarge the domain of attraction of a linear system subject to nested saturation. In \cite{WPWZ2014}, both static and dynamic anti-windup problems were investigated for LFRs of rational systems with input saturation. However, anti-windup compensation is typically effective only when the control input remains within a relatively moderate range. Second, classical absolute-stability tools, such as the circle criterion, the Popov criterion, and LMI-based conditions, have also been employed to establish regional stability for linear systems with input constraints; see \cite{HB1998,K1999,KI2000}. Nevertheless, the Popov criterion is mainly used for stability analysis, while circle-criterion-based synthesis may be conservative.

% {\color{blue}{Third, various globally and semi-globally stabilizing feedback laws, including both state-feedback and output-feedback designs, have been developed in \cite{RJJ1997,HEY1994,FZQ2007}. Moreover, a gain-scheduling output-feedback saturation controller was proposed in \cite{HZT2004} to attenuate disturbances in addition to ensuring local stability. Based on norm-bounded differential inclusion models, two gain-scheduling output-feedback controllers were further proposed for linear time-invariant (LTI) and linear parameter-varying (LPV) systems with input saturation in \cite{XF2015,XF2016}, respectively, where regional stability for unstable plants was also achieved.}}

The IQC framework has emerged as a powerful tool for characterizing a broad class of uncertainties and nonlinearities, including time-invariant and time-varying parametric uncertainties, norm-bounded uncertainties, as well as sector-bounded and slope-restricted nonlinearities. This framework was originally developed by Megretski and Rantzer \cite{AA1997}. Early developments in IQC theory mainly focused on robust stability and performance analysis, and many important results have been established for linear time-invariant and time-varying systems. For instance, a necessary and sufficient condition for robust stability with dynamic IQCs was established in \cite{CI2008}. In \cite{JCH2015}, a systematic procedure for stability and performance analysis of a broad class of uncertain systems was developed based on dynamic IQCs, with the corresponding conditions formulated as LMIs. However, dynamic-IQC-based analysis usually requires a fixed basis function for the inner approximation of multipliers. If the basis dimension is chosen to be high, the resulting conditions become computationally demanding; if it is chosen to be low, additional conservatism may be introduced.

On the other hand, motivated by the strong descriptive capability of IQCs for both uncertainties and nonlinearities, increasing attention has been paid to IQC-based controller synthesis. In particular, Seiler \cite{P2015} developed a synthesis-oriented condition by combining an IQC-based time-domain dissipation inequality with a quadratic Lyapunov function, which has since become an effective tool for robust controller design. In this context, $J$-spectral factorization also plays an important role in synthesis and will be used in this paper. Subsequently, Wang \textit{et al.} \cite{WPS2016} proposed a robust synthesis algorithm for a class of uncertain LPV systems, where an LPV synthesis step and an IQC analysis step are carried out iteratively. Yuan and Wu proposed an exact-memory-delay design framework for a class of uncertain systems with state and input delays, and both state-feedback and output-feedback $\mathcal{H}_\infty$ controllers based on dissipation inequalities were studied in \cite{CF2016a,CF2016,CF2017}. In addition, Coutinho \cite{DM2002} presented a method for the optimal control of LTI systems with input saturation. However, robust synthesis results for uncertain LFT systems with input saturation remain limited, especially when dead-zone nonlinearities and structured time-varying uncertainties are to be treated simultaneously within a unified IQC framework while retaining computational tractability.

Moreover, recent work has highlighted the versatility of IQCs as an analysis tool for nonlinear and uncertain feedback interconnections. In particular, IQC-based formulations have been used to compute inner estimates of regions of attraction via SOS optimization~\cite{Iannellia, Iannellib} and to assess worst-case performance bounds, including settings with saturation and data-driven/local IQC descriptions~\cite{Pusch2022}. However, these developments are primarily analysis-oriented and do not directly yield tractable controller synthesis conditions for uncertain LFT systems with input saturation.

\textit{Contributions.} In the spirit of \cite{CF2016}, this paper studies robust $\mathcal{H}_\infty$ control synthesis for a class of uncertain systems with input saturation within the IQC framework. The uncertain plant is represented as an LFT model involving block-structured time-varying uncertainties and a dead-zone nonlinearity. A loop transformation is first introduced so that the Popov IQC can be incorporated in a proper form for state-feedback synthesis. Then, each element of the bounded time-varying uncertainty set is characterized by a static IQC, while the dead-zone nonlinearity is described by mixed IQCs, including the Popov multiplier, the Zames--Falb multiplier, and the sector-bound multiplier with scaling factors. Based on these descriptions, a new scaled bounded real lemma is derived for closed-loop robust stability and performance analysis. Although it is reminiscent of the scaled bounded real lemma in \cite{PP1995}, the present result is developed for synthesis and explicitly accommodates multiple IQCs for each time-varying uncertain element. By feeding the dead-zone output into the state-feedback law, robust closed-loop stability is guaranteed and an $\mathcal{L}_2$-gain performance from the disturbance input to the error output is achieved. The resulting synthesis conditions are expressed as a finite set of LMIs, which become convex after standard variable substitutions and relaxations. The main contributions of this paper are summarized as follows:
\begin{itemize}
	\item [1)] An $\mathcal{H}_\infty$ state-feedback control law based on mixed IQCs is developed for a class of uncertain LFT systems with input saturation.
	
	\item [2)] A new scaled bounded real lemma is established for controller synthesis, where both uncertainties and nonlinearities in the closed-loop system are characterized in a unified IQC framework.
	
	\item [3)] Compared with approaches based on a single IQC or a sector-bound condition for the dead-zone nonlinearity, the proposed method achieves improved $\mathcal{L}_2$-gain performance by employing dynamic IQCs together with tunable scaling factors.
\end{itemize}

\textit{Paper organization.} The remainder of this paper is organized as follows. Section \uppercase\expandafter{\romannumeral2} presents the system representation and the required IQC preliminaries. Section \uppercase\expandafter{\romannumeral3} develops the main results: it first introduces a loop transformation to obtain a physically realizable LFT interconnection, then presents the mixed-IQC characterization of the uncertainty and dead-zone nonlinearity, and finally derives an $\mathcal{H}_\infty$ state-feedback controller synthesis condition based on a new scaled bounded real lemma. Section \uppercase\expandafter{\romannumeral4} illustrates the proposed approach on a second-order uncertain LFT example and a cart--spring--pendulum example via simulations. Finally, Section \uppercase\expandafter{\romannumeral5} concludes the paper.

\textit{Notations.} Throughout the paper, $\mathbb{R}$ and $\mathbb{C}$ denote the sets of real and complex numbers, respectively, and $\mathbb{R}_+$ denotes the set of positive real numbers. $\mathbb{R}^{m\times n}$ ($\mathbb{C}^{m\times n}$) denotes the set of real (complex) $m\times n$ matrices, while $\mathbb{R}^{n}$ ($\mathbb{C}^{n}$) denotes the set of real (complex) $n$-dimensional vectors. The $n\times n$ identity matrix is denoted by $I_n$. We use $\mathbb{S}^{n}$ and $\mathbb{S}_+^{n}$ to denote the sets of real symmetric and positive definite $n\times n$ matrices, respectively. In LMIs, the symbol $\star$ denotes the terms induced by symmetry. For two integers $k_1<k_2$, let $\mathbf{I}[k_1,k_2]\triangleq\{k_1,k_1+1,\ldots,k_2\}$. For $s\in\mathbb{C}$, $\bar{s}$ denotes the complex conjugate of $s$. For a matrix $M\in\mathbb{C}^{m\times n}$, $M^{\rm T}$ denotes its transpose and $M^{*}$ denotes its conjugate transpose. Let $\mathbb{RL}_\infty$ denote the set of proper rational functions with real coefficients and no poles on the imaginary axis, and let $\mathbb{RH}_\infty$ denote the subset of $\mathbb{RL}_\infty$ consisting of functions that are analytic in the closed right-half complex plane. Similarly, $\mathbb{RL}_\infty^{m\times n}$ and $\mathbb{RH}_\infty^{m\times n}$ denote the sets of $m\times n$ transfer matrices whose entries belong to $\mathbb{RL}_\infty$ and $\mathbb{RH}_\infty$, respectively. For $G\in\mathbb{RL}_\infty^{m\times n}$, its para-Hermitian conjugate, denoted by $G^{\sim}$, is defined as $G^{\sim}(s)\triangleq G(-\bar{s})^{*}$. For $x\in\mathbb{C}^{n}$, the Euclidean norm is defined by $\|x\|\triangleq(x^{*}x)^{1/2}$. Finally, $L_{2+}^{n}$ denotes the space of functions $u:[0,\infty)\to\mathbb{R}^{n}$ satisfying
$\|u\|_2\triangleq\left(\int_0^\infty u^{\rm T}(t)u(t)\,dt\right)^{1/2}<\infty$.

\section{Problem Formulation and Preliminaries}
\subsection{IQC preliminaries}
Static and dynamic IQCs will be used to characterize the input-output behavior of the uncertainties and nonlinearities considered in this paper. Therefore, this section briefly reviews several preliminary results on IQCs.
\begin{Def}\cite{P2015}
	Let  $\Pi\in\mathbb{RL}_\infty^{(m_1+m_2)\times(m_1+m_2)}$ be a proper and rational function, called a multiplier, such that $\Pi=\Psi^{\sim}W\Psi$ with $W\in\mathbb{R}^{n_z\times n_z}$ and $\Psi\in\mathbb{RH}_\infty^{n_z\times(m_1+m_2)}$. Then two signals $v\in L_{2e+}^{m_1}$ and $w\in L_{2e+}^{m_2}$ satisfy the IQC defined by the multiplier $\Pi$, and $(\Psi, W)$ is a hard IQC factorization of $\Pi$ if for any bounded, causal operator $\mathcal{S}$ satisfying the IQC defined by $\Pi$, the following inequality holds:
	\begin{eqnarray}\label{eq:4}
	\int_0^{T}z^{\rm T}(t)Wz(t)dt\ge 0
	\end{eqnarray}
	where $z\in\mathbb{R}^{n_z}$ denotes the filtered output of $\Psi$ driven by inputs $(v, w)$ with zero initial conditions, i.e., $z = \Psi\left[\begin{array}{c}
	v \\
	w
	\end{array}\right]$, $w=\mathcal{S}(v)$, and all $T\ge 0$.
\end{Def}
To synthesize robust controller within IQC framework, the lemma of dissipation inequality is provided below.
\begin{Lem}\label{Lem:1} \cite{P2015}
	Consider a linear system $G\in\mathbb{RH}_\infty^{n_y\times n_u}$ and a bounded causal operator $\Delta\in\mathbb{R}^{n_q\times n_q}$. Assume that:
	\begin{enumerate}[1)]
		\item the interconnection of $G$ and $\Delta$ is well-posed.
		
		\item $\Delta$ satisfies the IQC defined by $\Pi$ and $(\Psi,W)$ is a hard factorization of $\Pi$.
		
		\item there exists $P\ge 0$ and a scalar $\gamma>0$ such that $V(x)=x^{\rm T}Px$ satisfies
		\begin{eqnarray}\label{eq:5}
		z^{\rm T}Wz+\dot{V}<\gamma d^{\rm T}d-\frac{1}{\gamma} e^{\rm T} e
		\end{eqnarray}
		for all nontrivial $(x,d,e)\in\mathbb{R}^{n_x+n_\psi}\times\mathbb{R}^{n_d}\times\mathbb{R}^{n_e}$.
	\end{enumerate}
	Then, the feedback interconnection of $G$ and $\Delta$ is stable.
\end{Lem}

Note that a hard IQC factorization, together with the following $J$-spectral factorization result, is essential for applying the dissipation theorem. In addition to $J$-spectral factorizations, a multiplier $\Pi$ may admit other factorizations; see \cite{JWA2012,B1987,G1995} for further discussions. Thus, the hard property is not intrinsic to the multiplier itself, but depends on the particular factorization $(\Psi,W)$. In \cite{AA1997}, a rich collection of IQCs is presented for characterizing various classes of uncertainties and nonlinearities.

\begin{Def}\cite{P2015}
	$(\Psi,W)$ is called a $J_{m_1,m_2}$-spectral factor of $\Pi=\Pi^{\sim}\in{\mathbb{RL}_\infty^{(m_1+m_2)\times(m_1+m_2)}}$ if $\Pi = \Psi^{\sim}W\Psi$, $W=\left[\begin{array}{cc}
	I_{m_1} & 0 \\
	0 & -I_{m_2}
	\end{array}\right]$, and $\Psi,\Psi^{-1}\in{\mathbb{RH}_\infty^{(m_1+m_2)\times(m_1+m_2)}}$.
\end{Def}
\begin{Lem}\cite{P2015}{\label{Lem:2}}
	Let $\Pi = \Pi^{\sim}\in\mathbb{RL}_\infty^{(m_1+m_2)\times(m_1+m_2)}$ and partition this multiplier as a block matrix, i.e., $\Pi=\left[\begin{array}{cc}
	\Pi_{11} & \Pi_{12} \\
	\Pi_{12}^{\sim} & \Pi_{22}
	\end{array}\right]$, where $\Pi_{11}\in{\mathbb{RL}_{\infty}^{m_1\times m_1}}$ and $\Pi_{22}\in{\mathbb{RL}_{\infty}^{m_2\times m_2}}$. If $\Pi_{11}(j\omega)>0$ and $\Pi_{22}(j\omega)<0$ $\forall{\omega}\in{\mathbb{R}}\bigcup\{\infty\}$, then
	\begin{itemize}
		\item [1)] $\Pi$ has a $J_{m_1,m_2}$-spectral factorization $(\Psi,W)$;
		\item [2)] The $J_{m_1,m_2}$-spectral factorization $(\Psi,W)$ is a hard factorization of $\Pi$.
	\end{itemize}
\end{Lem}

Lemma \ref{Lem:2} provides a sufficient condition for the existence of a $J_{m_1,m_2}$-spectral factorization. This result will be used later to construct hard IQC factorizations for the multipliers characterizing the dead-zone nonlinearity in the controller synthesis procedure.

\subsection{System description and control objective}
This section presents the problem formulation for robust stabilization and performance analysis of a class of uncertain systems subject to input saturation within a general $\mathcal{H}_\infty$ framework.

Consider a class of uncertain linear systems subject to input saturation, represented in the following LFT form:
\begin{eqnarray}\label{eq:1}
\dot{x}_p &\hspace*{-1.5ex}=\hspace*{-1.5ex}& Ax_p + B_1p + B_2d + B_0Sat(u), \nonumber \\
q &\hspace*{-1.5ex}=\hspace*{-1.5ex}& C_0x_p + D_{01}p + D_{02}d + D_{00}Sat(u), \nonumber \\
e &\hspace*{-1.5ex}=\hspace*{-1.5ex}& C_1x_p + D_{11}p + D_{12}d + D_{10}Sat(u), \nonumber \\
p &\hspace*{-1.5ex}=\hspace*{-1.5ex}& \Delta q,
\end{eqnarray}
where $x_p\in\mathbb{R}^{n_x}$ is the state vector of the plant, $u\in\mathbb{R}^{n_u}$ is the control input, $d\in\mathbb{R}^{n_d}$ is the exogenous input including external disturbances, measurement noise, and reference signals, and $e\in\mathbb{R}^{n_e}$ is the performance output. The signals $p,q\in\mathbb{R}^{n_q}$ are the internal variables interconnected with the uncertainty block, representing its input and output, respectively. The unknown matrix $\Delta$ denotes a norm-bounded, time-varying structured uncertainty of the form
\begin{align}\label{eq:2}
\bm{\Delta} &\triangleq \Big\{\mathrm{diag}\{\delta_1I_{m_1},\ldots,\delta_{\nu}I_{m_{\nu}},\Delta_{{\nu}+1},\ldots,\Delta_{{\nu}+f}\}: \nonumber \\
& \delta_i:\mathbb{R}_+\to \mathbb{R},\ |\delta_i|\le b,\ i\in\mathbf{I}[1,{\nu}]; \nonumber \\
& \Delta_{{\nu}+j}:\mathbb{R}_+\to \mathbb{R}^{r_j\times r_j},\ \|\Delta_{{\nu}+j}\|\le b,\ j\in\mathbf{I}[1,f]\Big\},
\end{align}
where $b$ is a positive constant, and $\sum_{i=1}^{{\nu}}m_i+\sum_{j=1}^{f}r_j=n_q$.

From \eqref{eq:1}, the saturation nonlinearity $Sat(\cdot)$ acts componentwise on $u$. Without loss of generality, the input saturation nonlinearity for $n_u=1$ is described by
\begin{eqnarray}\label{eq:3}
Sat(u)\triangleq
\left\{
\begin{array}{ll}
\bar{u}\operatorname{sign}(u), & |u|\ge \bar{u}, \\[0.5ex]
u, & |u|<\bar{u},
\end{array}
\right.
\end{eqnarray}
where $\bar{u}>0$ denotes the saturation bound. It is clear that $Sat(u)$ is not differentiable at $|u|=\bar{u}$. 

We impose the following standard assumption on the nominal plant.

\begin{Asu}\label{Asu:2}
	The state variables of non-saturated linear system (1) is measurable, $A$ is stable and the pair $(A,[B_1~B_0])$ is stabilizable.
\end{Asu}

Note that Assumption \ref{Asu:2} is necessary to guarantee
the existence of a feedback stabilizing controller for a non-saturated uncertain linear system (\ref{eq:1}).

In this paper, we consider the robust $\mathcal{H}_\infty$ state-feedback control problem for the uncertain system \eqref{eq:1} within the IQC framework. The objective is to design a state-feedback control law such that, for all admissible uncertainties $\Delta\in\boldsymbol{\Delta}$ and input saturation nonlinearities $Sat(u)$, the resulting closed-loop system is stable and the $\mathcal{L}_2$ gain from the disturbance input $d$ to the performance output $e$ is bounded by a scalar $\gamma$, namely,
$\|e\|_2 \le \gamma \|d\|_2$ under zero initial conditions.

%\begin{Def}\cite{MP1996}
%    A subset $\mathcal{D}$ of the complex plane is called an LMI region if there exist a symmetric matrix $\xi=[\xi_{kl}]\in\mathbb{R}^{m\times m}$ and a matrix $\eta_{kl}\in\mathbb{R}^{m\times m}$ such that
%    \begin{eqnarray}
%        \mathcal{D} = \{y\in\mathbb{C}:f_{\mathcal{D}}(y)<0\}
%    \end{eqnarray}
%    with
%    $f_{\mathcal{D}}(y):=\xi+y\eta+\bar{y}\eta^{\rm T}=[\xi_{kl}+\eta_{kl}y+\eta_{lk}\bar{z}]_{1\le k,l\le m}$.
%\end{Def}
%\begin{Lem}\cite{MP1996}
%    The matrix $A$ is $\mathcal{D}$-stable if and only if there exists a symmetric matrix $X_D$ such that
%    \begin{eqnarray}
%        M_{D}(A,X_D)<0,~~X_D>0
%    \end{eqnarray}
%\end{Lem}
\section{Main Results}
This section develops the main controller synthesis results. We first introduce a loop transformation to obtain a physically realizable representation of the closed-loop system. The transformed system is then expressed as a standard LTI system interconnected with the uncertainty block $\Delta$ and the nonlinearity block $\mathcal{N}$. Based on this representation, mixed IQCs are incorporated to derive state-feedback controller synthesis conditions.

\subsection{System Transformation}

As described in the previous section, the LFR model in \eqref{eq:1} contains both block-structured time-varying uncertainties and an input saturation nonlinearity acting on the LTI system. To facilitate controller synthesis within a unified IQC framework, we first reformulate the input saturation by separating the linear and nonlinear parts. Specifically, the saturation nonlinearity is written as
\[
Sat(u(t)) = u(t)-\mathcal{N}(u(t)),
\]
where $\mathcal{N}(u(t))$ denotes the dead-zone nonlinearity.

To characterize the dead-zone nonlinearity by dynamic IQCs, especially the Popov IQC, a loop transformation is required. As discussed in \cite{AA1997,HP2014}, the Popov multiplier is not proper and hence cannot be directly used in the present synthesis framework. To overcome this difficulty, we introduce the loop transformation
\[
\bar{\Delta}_{\mathcal{N}}\triangleq\Delta_{\mathcal{N}}\circ \frac{1}{s+\alpha},
\]
that is, $w(t)=(\bar{\Delta}_{\mathcal{N}}v)(t)$, where
\[
\dot{u}(t)=-\alpha u(t)+v(t), \qquad u(0)=0,
\]
with $\alpha>0$. This modification is illustrated in Fig.~\ref{Fig:1}. 
\begin{figure}[htbp]
	\centering
	\includegraphics[width=0.45\textwidth]{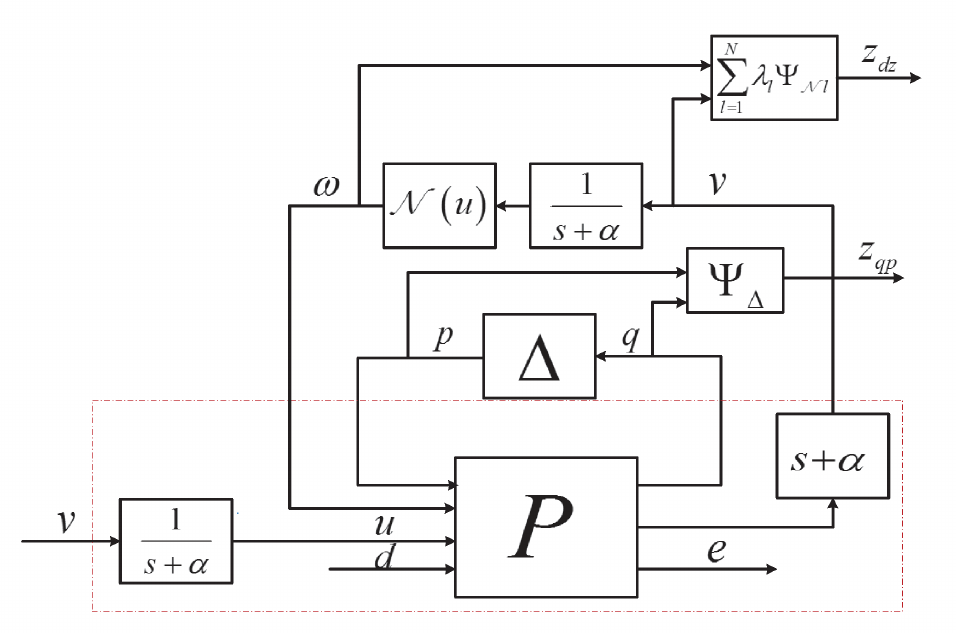}
	\caption{Transformed uncertain LFT system. \label{Fig:1}}
\end{figure}

To ensure physical realizability, the factor $(s+\alpha)$ introduced by the loop transformation must be absorbed into the known plant dynamics, which motivates the following assumption.

\begin{Asu}\label{Asu:1}
	The nominal part of the LFT system \eqref{eq:1} is strictly proper.
\end{Asu}

With the above loop transformation, the original uncertain LFT system can be rewritten as a new LFT interconnection with two layers of uncertain/nonlinear blocks, namely, the bounded time-varying structured uncertainty block and the saturation-induced dead-zone nonlinearity block. In particular, the inputs and outputs of these two blocks are fed into the corresponding IQC multipliers. The transformed LFT representation is depicted in Fig.~\ref{Fig:1}. Absorbing the factor $(s+\alpha)$ into the known part of the plant yields
\begin{eqnarray}\label{eq:12}
\left[\begin{array}{c}
\dot{x}_p \\
\dot{u} \\
q \\
e
\end{array}\right]&\hspace*{-1.5ex}=\hspace*{-1.5ex}&\left[\begin{array}{cccc}
A & B_0 & B_1 & -B_0 \\
0 & -\alpha I_{n_u} & 0 & 0 \\
C_0 & D_{00} & D_{01} & -D_{00} \\
C_1 & D_{10} & D_{11} & -D_{10} \\
\end{array}\right. \nonumber \\
&\hspace*{-1.5ex} \hspace*{-1.5ex}&\left.\begin{array}{cc}
B_2 & 0 \\
0 & I_{n_u} \\
D_{02} & 0 \\
D_{12} & 0
\end{array}\right]\left[\begin{array}{c}
x_p \\
u \\
p \\
w \\
d \\
v
\end{array}\right] \nonumber \\
p &\hspace*{-1.5ex}=\hspace*{-1.5ex}& \Delta q \nonumber \\
\mathcal{N}(u) &\hspace*{-1.5ex}=\hspace*{-1.5ex}& u-Sat(u),
\end{eqnarray}
where the augmented state vector is $x=[x_p^{\rm T}~u^{\rm T}]^{\rm T}\in\mathbb{R}^{n_x+n_u}$, and $v\in\mathbb{R}^{n_u}$ is the newly introduced input associated with the loop transformation.

\subsection{IQC characterization}
We next characterize the dead-zone nonlinearity and the structured uncertainty block by IQCs. To begin with, the input saturation $Sat(u)$ is a memoryless nonlinearity in the sector $[0,1]$; see \cite{HP2014}. Hence,
\begin{eqnarray}\label{eq:6}
Sat(u)\big(u-Sat(u)\big)\ge 0,\qquad \forall u\in\mathbb{R},
\end{eqnarray}
which implies a sector-bound condition. This condition can be characterized by the static IQC multiplier
\[
\Pi_S=\left[\begin{array}{cc}
0.01 & 1 \\
1 & -2.01
\end{array}\right].
\]
Here, a small positive perturbation is introduced so that the multiplier satisfies the strict definiteness condition required by Lemma~\ref{Lem:2}. Moreover, if uncertainties or nonlinearities satisfy several IQCs defined by multipliers $\Pi_1,\ldots,\Pi_n$, then any positive linear combination
\[
\Pi=\lambda_1\Pi_1+\cdots+\lambda_n\Pi_n, \qquad \lambda_i>0,
\]
also defines a valid IQC. In this sense, $\Pi_S$ can be viewed as a combination of
\[
\left[\begin{array}{cc}
0 & 1 \\
1 & -2
\end{array}\right]
\quad \text{and} \quad
0.01\left[\begin{array}{cc}
1 & 0\\
0 & -1
\end{array}\right],
\]
which ensures that the resulting multiplier satisfies the condition of Lemma~\ref{Lem:2} and thus admits a $J$-spectral factorization.

Besides the static IQC, the Popov IQC can also be used to describe sector-bounded nonlinearities. In particular, the Popov multiplier
\[
\Pi_P\triangleq\pm\left[\begin{array}{cc}
0 & -j\omega \\
j\omega & 0
\end{array}\right]
\]
provides an alternative characterization. However, this multiplier is not proper because it is unbounded on the imaginary axis. Owing to the above loop transformation, a proper Popov-type IQC can be constructed for the transformed nonlinearity $\bar{\Delta}_{\mathcal{N}}$. Furthermore, to ensure a hard IQC factorization and satisfy Lemma~\ref{Lem:2}, we consider the slightly modified Popov multiplier
\[
\Pi_P = \left[\begin{array}{cc}
0 & -s \\
s & 0
\end{array}\right]
+0.01\left[\begin{array}{cc}
1 & 0 \\
0 & -1
\end{array}\right].
\]
Accordingly, the transformed nonlinearity $\bar{\Delta}_{\mathcal{N}}$ can be characterized by
\begin{eqnarray}\label{eq:7}
\Pi_{\bar{P}} = \left[\begin{array}{cc}
\frac{1}{s+\alpha} & 0\\
0 & 1
\end{array}\right]^{\sim}\left[\begin{array}{cc}
0.01 & -s \\
s & -0.01
\end{array}\right]\left[\begin{array}{cc}
\frac{1}{s+\alpha} & 0 \\
0 & 1
\end{array}\right].
\end{eqnarray}
This modified Popov multiplier satisfies the conditions of Lemma~\ref{Lem:2}, and hence admits a hard $J$-spectral factorization
\[
\Pi_{\bar{P}} = \Psi_{\bar{P}}^{\sim} W_{\bar{P}} \Psi_{\bar{P}},
\]
which will be used in the subsequent dissipation-inequality-based synthesis.

Another useful characterization is provided by the Zames--Falb IQC, which applies to odd and monotone nonlinearities; see \cite{GP1968,DM2011} and the references therein. The corresponding multiplier is given by
\[
\Pi_{ZF} \triangleq \left[\begin{array}{cc}
0.01 & 1+H(s) \\
1+H(s)^{*} & -2.01-(H(s)+H(s)^{*})
\end{array}\right],
\]
where $H\in\mathbb{RL}_\infty$ is arbitrary except that the $\mathcal{L}_1$ norm of its impulse response is no larger than $1$.

In view of Fig.~\ref{Fig:1}, we combine the static IQC, the Popov IQC, and the Zames--Falb IQC to characterize the sector-bounded dead-zone nonlinearity. Similar to the Popov case, the static and Zames--Falb IQC multipliers are adapted to the transformed nonlinearity $\bar{\Delta}_{\mathcal{N}}$ as
\begin{eqnarray}\label{eq:8}
\Pi_{\bar{S}} = \left[\begin{array}{cc}
\frac{1}{s+\alpha} & 0 \\
0 & 1
\end{array}\right]^{\sim}\left[\begin{array}{cc}
0.01 & 1 \\
1 & -2.01
\end{array}\right]\left[\begin{array}{cc}
\frac{1}{s+\alpha} & 0 \\
0 & 1
\end{array}\right]
\end{eqnarray}
and
\begin{eqnarray}\label{eq:9}
&\hspace*{-1.5ex} \hspace*{-1.5ex}& \Pi_{\bar{ZF}} = \left[\begin{array}{cc}
\frac{1}{s+\alpha} & 0 \\
0 & 1
\end{array}\right]^{\sim}\nonumber \\
&\hspace*{-1.5ex} \hspace*{-1.5ex}&\left[\begin{array}{cc}
0.01 & 1+H(s) \\
1+H(s)^{*} & -2.01-(H(s)+H(s)^{*})
\end{array}\right]\left[\begin{array}{cc}
\frac{1}{s+\alpha} & 0 \\
0 & 1
\end{array}\right]. \nonumber \\
&\hspace*{-1.5ex} \hspace*{-1.5ex}&
\end{eqnarray}
To balance the roles of the static, Popov, and Zames--Falb multipliers in characterizing the sector-bounded nonlinearity, we introduce positive scaling factors $\lambda_i>0$, $i=1,2,3$, and define the mixed IQC multiplier as
\[
\Pi_{\mathcal{N}}= \lambda_1\Pi_{\bar{P}}+\lambda_2\Pi_{\bar{ZF}}+\lambda_3\Pi_{\bar{S}}.
\]

\begin{Rem}\label{Rem:1}
	The Zames--Falb IQC multiplier itself is proper if a suitable $H(s)$ is chosen. Therefore, strictly speaking, it is not necessary to introduce an inertial element when only the Zames--Falb IQC is used to characterize the dead-zone nonlinearity. If the Zames--Falb IQC is considered without loop transformation, the explicit form of the transformed LFT structure would be slightly different, whereas the robust controller synthesis procedure would remain essentially the same. Due to space limitations, the details are omitted.
\end{Rem}

We now turn to the IQC characterization of the time-varying uncertainty block $\Delta$. From \eqref{eq:2}, each diagonal block of $\bm{\Delta}$ is norm-bounded by a positive scalar $b$, that is, $\|\Delta\|_2\le b$. Introduce scaling matrices $X_{m_k}\in\mathbb{S}^{m_k\times m_k}$ and scalars $\chi_k$. Then, after $J_{n_q,n_q}$-spectral factorization, the IQC multiplier associated with the repeated scalar block can be written as follows.
For $k\in\mathbf{I}[1,{\nu}]$, let
\[
{\Psi}_{\Delta_k}\triangleq\left[\begin{array}{cc}
bI_{m_k} & 0 \\
0 & I_{m_k}
\end{array}\right], \quad
{W}_{\Delta,k}\triangleq\left[\begin{array}{cc}
X_{m_k} & 0 \\
0 & -X_{m_k}
\end{array}\right],
\]
where $X_{m_k}\in\mathbb{S}^{m_k\times m_k}$. Then the corresponding IQC multiplier is
\begin{align}\label{eq:10}
&\Pi_{\Delta_k}
\triangleq
\Psi_{\Delta_k}^{\sim}W_{\Delta,k}\Psi_{\Delta_k} \nonumber \\
&=
\left[\begin{array}{cc}
bI_{m_k} & 0 \\
0 & I_{m_k}
\end{array}\right]^{\sim}
\left[\begin{array}{cc}
X_{m_k} & 0 \\
0 & -X_{m_k}
\end{array}\right]
\left[\begin{array}{cc}
bI_{m_k} & 0 \\
0 & I_{m_k}
\end{array}\right].
\end{align}

Similarly, for the remaining uncertainty blocks, let
\[
{\Psi}_{\Delta_k}\triangleq\left[\begin{array}{cc}
b & 0 \\
0 & 1
\end{array}\right], \qquad
{W}_{\Delta,k}\triangleq\left[\begin{array}{cc}
\chi_k & 0 \\
0 & -\chi_k
\end{array}\right],
\]
for $k\in\mathbf{I}[{\nu}+1,n_q]$, where $\chi_k>0$. Then the corresponding IQC multiplier is
\begin{align}\label{eq:11}
\Pi_{\Delta_k}
&\triangleq
\Psi_{\Delta_k}^{\sim}W_{\Delta,k}\Psi_{\Delta_k} \nonumber \\
&=
\left[\begin{array}{cc}
b & 0 \\
0 & 1
\end{array}\right]^{\sim}
\left[\begin{array}{cc}
\chi_k & 0 \\
0 & -\chi_k
\end{array}\right]
\left[\begin{array}{cc}
b & 0 \\
0 & 1
\end{array}\right].
\end{align}

\subsection{Mixed IQC-Based State-Feedback Controller Synthesis}

With the loop transformation and the above IQC characterizations, the robust $\mathcal{H}_\infty$ control problem for the original input-saturated uncertain system \eqref{eq:1} is converted into a controller synthesis problem for the transformed system \eqref{eq:12} within the mixed-IQC framework. Specifically, the objective is to design a $\gamma$-suboptimal state-feedback controller for \eqref{eq:12} such that the resulting closed-loop system is stable and satisfies $\|e\|_2<\gamma\|d\|_2$
under zero initial conditions for all admissible uncertainties and saturation nonlinearities.

For an IQC multiplier $\Pi=\Pi^{\sim}\in\mathbb{RL}_\infty^{(m_1+m_2)\times(m_1+m_2)}$, partitioned as $\Pi=\begin{bmatrix}
\Pi_{11} & \Pi_{12}\\
\Pi_{12}^{\sim} & \Pi_{22}
\end{bmatrix}$, where $\Pi_{11}$ and $\Pi_{22}$ are of dimensions $m_1\times m_1$ and $m_2\times m_2$, respectively, suppose that $\Pi_{11}(j\omega)>0$ and $\Pi_{22}(j\omega)<0$ for all $\omega\in\mathbb{R}\cup\{\infty\}$. Then $\Pi$ admits a $J_{m_1,m_2}$-spectral factorization $(\Psi,W)$, where
\[
\Psi\triangleq\begin{bmatrix}
\Psi_{11} & \Psi_{12}\\
\Psi_{21} & \Psi_{22}
\end{bmatrix}, \qquad
W\triangleq\begin{bmatrix}
I_{m_1} & 0\\
0 & -I_{m_2}
\end{bmatrix}.
\]
To facilitate the robust stability analysis and derive the performance inequality $\frac{1}{\gamma}\int_0^T e^{\rm T}(t)e(t)\,dt
<
\gamma\int_0^T d^{\rm T}(t)d(t)\,dt$, 
the factor $\Psi$ is further converted into the special form $\bar{\Psi}\triangleq
\begin{bmatrix}
\bar{\Psi}_{11} & \bar{\Psi}_{12}\\
0 & I_{m_2}
\end{bmatrix}$ where $\bar{\Psi}_{11}\triangleq\Psi_{11}-\Psi_{12}\Psi_{22}^{-1}\Psi_{21}$ and $\bar{\Psi}_{12}\triangleq\Psi_{12}\Psi_{22}^{-1}$. We refer the reader to \cite{CF2016,CF2017} for further details.

With this preparation, the IQC-induced system $\Psi_{\Delta_k}$ associated with the time-varying structured uncertainty block $\Delta$ in \eqref{eq:12} can be written, for all $k\in\mathbf{I}[1,n_q]$, as
\begin{eqnarray}\label{eq:13}
\begin{bmatrix}
z_{\Delta 1,k} \\
z_{\Delta 2,k}
\end{bmatrix}
=
\begin{bmatrix}
D_{\Delta\bar{\Psi}1,k} & D_{\Delta\bar{\Psi}2,k} \\
0 & I_{n_q}
\end{bmatrix}
\begin{bmatrix}
q \\
p
\end{bmatrix}.
\end{eqnarray}
For each bounded time-varying uncertainty block $\Delta_k\in\bm{\Delta}$, the corresponding multipliers \eqref{eq:10} and \eqref{eq:11} are used to characterize its input-output behavior for $k\in\mathbf{I}[1,{\nu}]$ and $k\in\mathbf{I}[{\nu}+1,n_q]$, respectively. Owing to the special structure of $\bar{\Psi}$ described above, the output $z_{\Delta,k}=\left[\begin{matrix}z_{\Delta 1,k}^{\rm T} & z_{\Delta 2,k}^{\rm T}\end{matrix}\right]^{\rm T}$ is given by \eqref{eq:13}.

Similarly, the IQC-induced LTI system $\Psi_{\mathcal{N},l}$ associated with the sector-bounded nonlinearity $\mathcal{N}$ in \eqref{eq:12} can be represented in state-space form, for $l\in\mathbf{I}[1,N_{n_\psi}]$ \footnote{We use static IQC, Popov IQC and Zames-Falb IQC to characterize the input-output behavior of dead-zone nonlinearity in our design, so the number here is $N_{n_\psi}=3$.}, as
\begin{eqnarray}\label{eq:14}
\begin{bmatrix}
\dot{x}_{\mathcal{N}\bar{\Psi}} \\
z_{\mathcal{N} 1,l} \\
z_{\mathcal{N} 2,l}
\end{bmatrix}
=
\begin{bmatrix}
A_{\mathcal{N}\bar{\Psi}} & B_{\mathcal{N}\bar{\Psi}1} & B_{\mathcal{N}\bar{\Psi}2} \\
C_{\mathcal{N}\bar{\Psi},l} & D_{\mathcal{N}\bar{\Psi}1,l} & D_{\mathcal{N}\bar{\Psi}2,l} \\
0 & 0 & I_{n_u}
\end{bmatrix}
\begin{bmatrix}
x_{\mathcal{N}\bar{\Psi}} \\
v\\
w
\end{bmatrix},
\end{eqnarray}
where the augmented state vector is $x_{\mathcal{N}\bar{\Psi}}\in\mathbb{R}^{n_\psi}$, and $n_\psi$ depends on the number of IQC multipliers used in the design. Moreover, $z_{\mathcal{N},l}=
\begin{bmatrix}
z_{\mathcal{N} 1,l}^{\rm T} & z_{\mathcal{N} 2,l}^{\rm T}
\end{bmatrix}^{\rm T}$ denotes the corresponding output vector.

It is worth emphasizing that different IQC strategies are adopted for the two classes of uncertain blocks. For the sector-bounded nonlinearity, several IQC multipliers with scaling factors are combined to obtain a mixed-IQC description, and the total number of such multipliers is denoted by $N_{n_\psi}$. In principle, more IQCs may be incorporated to refine the nonlinearity description, provided that they satisfy the corresponding definitions and factorization conditions. Since static IQCs can be viewed as a special case of dynamic IQCs, the use of dynamic IQCs generally provides a richer characterization of the nonlinearity. In particular, the introduction of scaling factors offers additional degrees of freedom for reducing the achievable $\mathcal{L}_2$ gain level $\gamma$.

By contrast, for the bounded time-varying uncertainty set $\bm{\Delta}$, we use one static IQC for each uncertainty block. As a result, the number of IQC multipliers is consistent with the number of blocks in $\bm{\Delta}$. Each multiplier is weighted by a scaling factor whose structure depends on the repeated scalar and full-block components. Based on this construction, we next derive a new scaled bounded real lemma within the IQC framework. According to the dissipation inequality, this result guarantees not only robust closed-loop stability but also the desired $\mathcal{L}_2$-gain performance. In this sense, the resulting condition plays a role similar to the standard scaled bounded real lemma (sBRL) \cite{PP1995}, but is developed here from an IQC-based perspective.

Based on Assumption \ref{Asu:2}, we consider the following state-feedback controller:
\begin{align}\label{eq:15}
v = F_c\begin{bmatrix}
x \\
x_{\mathcal{N}\bar{\Psi}}
\end{bmatrix} + H_cw,
\end{align}
where $F_c\in\mathbb{R}^{n_u\times(n_x+n_u+n_\psi)}$ and $H_c\in\mathbb{R}^{n_u\times n_u}$ are controller gain matrices to be designed. Unlike a standard state-feedback law, the dead-zone output $w$ is fed back as an additional scheduling variable together with the linear state-feedback term. This additional term plays an important role in adjusting the control signal online and improving the disturbance attenuation level $\gamma$.

Substituting the controller \eqref{eq:15} into \eqref{eq:12}, the augmented closed-loop system consisting of the transformed plant \eqref{eq:12} and the IQC-induced systems \eqref{eq:13}--\eqref{eq:14} can be written as
\begin{align}\label{eq:16}
\left[\begin{matrix}
\dot{x}_{cl} \\
z_{\Delta 1,k} \\
z_{\mathcal{N} 1,l} \\
e
\end{matrix}\right] &=\left[\begin{matrix}
A_{cl} & B_{cl0} & B_{cl1} & B_{cl2}\\
C_{cl\Delta1,k} & D_{cl\Delta10,k} & D_{cl\Delta11,k} & D_{cl\Delta12,k} \\
C_{cl\mathcal{N}1,l} & D_{cl\mathcal{N}10,l} & D_{cl\mathcal{N}11,l} & D_{cl\mathcal{N}12,l}\\
C_{cl2} & D_{cl20} & D_{cl21} & D_{cl22} 
\end{matrix}\right]\left[\begin{matrix}
x_{cl} \\
p \\
w \\
d
\end{matrix}\right],  \nonumber \\
p &= \Delta q, \nonumber \\
w &= \mathcal{N}(u),
\end{align}
where the overall state is $x_{cl}\triangleq[x^{\rm T}~ x_{\mathcal{N}\bar{\Psi}}^{\rm T}]^{\rm T}$ with
\begin{align}\label{eq:17}
A_{cl} &= \left[\begin{matrix}
A & B_0 & 0 \\
0 & -\alpha I_{n_u} & 0 \\
0 & 0 & A_{\mathcal{N}\bar{\Psi}}
\end{matrix}\right]+ \left[\begin{matrix}
0 \\
I_{n_u} \\
B_{\mathcal{N}\bar{\Psi}1}
\end{matrix}\right]F_c , \nonumber \\
B_{cl0} &= \left[\begin{matrix}
B_1 \\
0 \\
0
\end{matrix}\right], \nonumber \\
\left[\begin{matrix}
B_{cl1} & B_{cl2}
\end{matrix}\right] &=\left[\begin{array}{c:c}
\left[\begin{matrix}
-B_0 \\
0 \\
B_{\mathcal{N}\bar{\Psi}_2}
\end{matrix}\right] +\left[\begin{matrix}
0 \\
I_{n_u} \\
B_{\mathcal{N}\bar{\Psi}_1}
\end{matrix}\right]H_c & \left[\begin{matrix}
B_2 \\
0 \\
0
\end{matrix}\right]
\end{array}\right], \nonumber \\
\left[\begin{matrix}
C_{cl\Delta1,k} \\
C_{cl\mathcal{N}1,l} \\
C_{cl2}
\end{matrix}\right] &= \left[\begin{array}{c}
\left[\begin{matrix}
D_{\Delta\bar{\Psi}1,k}C_0 & D_{\Delta\bar{\Psi}1,k}D_{00} & 0
\end{matrix}\right]\\
\left[\begin{matrix}
0 & 0 & C_{\mathcal{N}\bar{\Psi},l}
\end{matrix}\right]+D_{\mathcal{N}\bar{\Psi}1,l}F_c \\
\left[\begin{matrix}
C_1 & D_{10} & 0
\end{matrix}\right]
\end{array}\right], \nonumber \\
D_{cl\Delta10,k} &= D_{\Delta\bar{\Psi}1,k}D_{01}+D_{\Delta\bar{\Psi}2,k}, D_{cl22} = D_{12}, \nonumber \\
D_{cl\Delta11,k} &= -D_{\Delta\bar{\Psi}1,k}D_{00}, D_{cl\Delta12,k} = D_{\Delta\bar{\Psi}1,k}D_{02}, \nonumber \\
D_{cl\mathcal{N}10,l} &= 0, D_{cl\mathcal{N}11,l} = D_{\mathcal{N}\bar{\Psi}1,l}H_c+D_{\mathcal{N}\bar{\Psi}2,l}, \nonumber \\
D_{cl\mathcal{N}12,l} &= 0,D_{cl20} = D_{11}, D_{cl21} = -D_{10}.
\end{align}

The subsequent LMI-based synthesis relies on a bounded-real-lemma-type condition. We therefore first establish a new scaled bounded real lemma for the robust $\mathcal{L}_2$ stability and performance analysis of the augmented closed-loop system \eqref{eq:16} within the IQC framework.
%In particular, for all $\textbf{I}[1,N_\lambda]$, we have
%    \begin{eqnarray}
%        C_{cl1,k} &\hspace*{-1.5ex}=\hspace*{-1.5ex}& \left[\begin{array}{c}
%          \bar{C}_{cl1,k} \\
%          0
%        \end{array}\right] \nonumber \\
%        &\hspace*{-1.5ex}=\hspace*{-1.5ex}& \left[\begin{array}{ccc}
%          0 & 0 & \bar{C}_{\Psi,k} \\
%          0 & 0 & 0
%        \end{array}\right] + \left[\begin{array}{c}
%          \bar{D}_{\Psi1,k} \\
%          0
%        \end{array}\right]F_c \nonumber
%        \end{eqnarray}
%        \begin{eqnarray}\left[
%        \begin{array}{ccc}
%          D_{cl10,k} & D_{cl11,k} & D_{cl12,k}\end{array}\right] \nonumber \\
%           = \left[\begin{array}{ccc}
%            0 & \bar{D}_{\Psi1,k}H_c+\bar{D}_{\Psi2,k} & 0 \\
%            0 & I_{n_u} & 0
%          \end{array}\right]
%    \end{eqnarray}

\begin{Lem}\label{Lem:3}
	Consider a time-varying structured parameter set $\bm{\Delta}$ defined in (\ref{eq:2}) and augmented closed-loop system (\ref{eq:16}). The following statements are equivalent.
	\begin{enumerate}[i)]
		\item The closed-loop system (\ref{eq:16}) with time-varying uncertainties and input saturation is robustly stable, and if there exist positive-definite block-diagonal scaling matrix $\Gamma=\text{diag}\{X_{m_1},\ldots,X_{m_s},\chi_{s+1},\ldots,\chi_{n_q}\}\in\mathbb{S}_{+}^{n_q}$, scaling factors $\lambda_l\in\mathbb{R}_{+}$ for all $l\in\textbf{I}[1,N_{n_\psi}]$ and a positive scalar $\gamma\in\mathbb{R}_+$, the $\mathcal{L}_2$-gain performance level is no larger than $\gamma$ for all bounded, time-varying structured uncertainties $\Delta\in{\bm{\Delta}}$ and dead-zone nonlinearity $\mathcal{N}(u)$.
		\item There exists a positive-definite matrix $P\in\mathbb{S}_+^{n_x+n_u+n_\psi}$ such that the matrix inequality \eqref{eq:18} holds, where the matrices are defined as follows.
		\begin{table*}
		\begin{align}\label{eq:18}
		\left[\begin{matrix}
		He\left\{PA_{cl}\right\} & \star & \star & \star & \star & \star & \star \\
		B_{cl0}^{\rm T}P & -\Gamma & \star & \star & \star & \star & \star \\
		B_{cl1}^{\rm T}P & 0 & -\sum_{l=1}^{N_{n_\psi}}\lambda_lI_{n_u} & \star & \star & \star & \star \\
		B_{cl2}^{\rm T}P & 0 & 0 & -\gamma I_{n_d}  & \star & \star & \star  \\
		\Theta_{cl\Delta 1} & \Theta_{cl\Delta 10} & \Theta_{cl\Delta 11} & \Theta_{cl\Delta12} & -\Gamma^{-1} & \star & \star \\
		\Xi_{cl\mathcal{N}1} & \Xi_{cl\mathcal{N}10} & \Xi_{cl\mathcal{N}11} & \Xi_{cl\mathcal{N}12} & 0 & -\Lambda & \star \\
		C_{cl2} & D_{cl20} & D_{cl21} & D_{cl22} & 0 & 0 & -\gamma I_{n_e}
		\end{matrix}\right]<0
		\end{align}
		\end{table*}
		\begin{align*}
		\Theta_{cl\Delta1}&= \left[\begin{matrix}
		C_{cl\Delta1,1}^{\rm T} & \ldots & C_{cl\Delta1,n_q}^{\rm T}
		\end{matrix}\right]^{\rm T} \nonumber \\
		\Theta_{cl\Delta10}&= \left[\begin{matrix}
		D_{cl\Delta10,1}^{\rm T} & \ldots & D_{cl\Delta10,n_q}^{\rm T}
		\end{matrix}\right]^{\rm T} \nonumber \\
		\Theta_{cl\Delta11}&= \left[\begin{matrix}
		D_{cl\Delta11,1}^{\rm T} & \ldots & D_{cl\Delta11,n_q}^{\rm T}
		\end{matrix}\right]^{\rm T} \nonumber \\
		\Theta_{cl\Delta12}&= \left[\begin{matrix}
		D_{cl\Delta12,1}^{\rm T} & \ldots & D_{cl\Delta12,n_q}^{\rm T}
		\end{matrix}\right]^{\rm T} \end{align*}
		\begin{align*}
		\Xi_{cl\mathcal{N}1} &= \left[\begin{matrix}
		C_{cl\mathcal{N}1,1}^{\rm T} & \ldots & C_{cl\mathcal{N}1,{N_{n_\psi}}}^{\rm T}
		\end{matrix}\right]^{\rm T} \nonumber \\
		\Xi_{cl\mathcal{N}10} &= \left[\begin{matrix}
		D_{cl\mathcal{N}10,1}^{\rm T} & \ldots & D_{cl\mathcal{N}10,{N_{n_\psi}}}^{\rm T}
		\end{matrix}\right]^{\rm T} \nonumber \\
		\Xi_{cl\mathcal{N}11} &= \left[\begin{matrix}
		D_{cl\mathcal{N}11,1}^{\rm T} & \ldots & D_{cl\mathcal{N}11,{N_{n_\psi}}}^{\rm T}
		\end{matrix}\right]^{\rm T} \nonumber \\
		\Xi_{cl\mathcal{N}12} &= \left[\begin{matrix}
		D_{cl\mathcal{N}12,1}^{\rm T} & \ldots & D_{cl\mathcal{N}12,{N_{n_\psi}}}^{\rm T}
		\end{matrix}\right]^{\rm T} \nonumber \\
		\Gamma &= \text{diag}\{X_{m_1},\ldots,X_{m_s},\chi_{s+1},\ldots,\chi_{n_q}\} \nonumber \\
		\Lambda &= \text{diag}\left\{\lambda_1^{-1},\ldots,\lambda_{N_{n_\psi}}^{-1}\right\}. \nonumber
		\end{align*}
	\end{enumerate}
\end{Lem}
\begin{proof}
	Considering a quadratic Lyapunov function $V(x_{cl}) = x_{cl}^{\rm T}Px_{cl}$ in terms of the augmented closed-loop system (\ref{eq:16}), and differentiating it with regard to time along the solution of system (\ref{eq:16}), then we have
	\begin{align}\label{eq:19}
	\dot{V} &= (A_{cl}x_{cl}+B_{cl0}p+B_{cl1}w+B_{cl2}d)^{\rm T}Px_{cl}\nonumber \\
	&+x_{cl}^{\rm T}P(A_{cl}x_{cl}+B_{cl0}p+B_{cl1}w+B_{cl2}d) \nonumber \\
	&= \left[\begin{matrix}
	x_{cl} \\
	p \\
	w \\
	d
	\end{matrix}\right]^{\rm T}\left[\begin{matrix}
	He\{PA_{cl}\} & \star & \star & \star \\
	B_{cl0}^{T} & 0 & \star & \star \\
	B_{cl1}^{T} & 0 & 0 & \star \\
	B_{cl2}^{T} & 0 & 0 & 0
	\end{matrix}\right]\left[\begin{matrix}
	x_{cl} \\
	p \\
	w \\
	d
	\end{matrix}\right].
	\end{align}
	Then given the inequality
	\begin{align}\label{eq:20}
	& \left[\begin{matrix}
	He\{PA_{cl}\} & \star & \star & \star \\
	B_{cl0}^{\rm T}P & -\Gamma & \star & \star \\
	B_{cl1}^{\rm T}P & 0 & -\sum_{i=1}^{N}\lambda_iI_{n_u} & \star \\
	B_{cl2}^{\rm T}P & 0 & 0 & -\gamma I_{n_d}
	\end{matrix}\right] \nonumber \\
	&+\sum_{k=1}^{s}\left[\begin{matrix}
	C_{cl\Delta1,k}^{\rm T} \\
	D_{cl\Delta10,k}^{\rm T} \\
	D_{cl\Delta11,k}^{\rm T} \\
	D_{cl\Delta12,k}^{\rm T}
	\end{matrix}\right] X_{m_k}\left[\begin{matrix}
	C_{cl\Delta1,k}^{\rm T}\\ D_{cl\Delta10,k}^{\rm T}\\ D_{cl\Delta11,k}^{\rm T}\\ D_{cl\Delta12,k}
	\end{matrix}\right]^{\rm T} \nonumber \\
	&+\sum_{k=s+1}^{N_q}\chi_k\left[\begin{matrix}
	C_{cl\Delta1,k}^{\rm T} \\
	D_{cl\Delta10,k}^{\rm T} \\
	D_{cl\Delta11,k}^{\rm T} \\
	D_{cl\Delta12,k}^{\rm T}
	\end{matrix}\right]\left[\begin{matrix}
	C_{cl\Delta1,k}^{\rm T}\\ D_{cl\Delta10,k}^{\rm T}\\ D_{cl\Delta11,k}^{\rm T}\\ D_{cl\Delta12,k}
	\end{matrix}\right]^{\rm T} \nonumber \\
	&+\sum_{l=1}^{N_{n_\psi}}\lambda_{l}\left[\begin{matrix}
	C_{cl\mathcal{N}1,l}^{\rm T} \\
	D_{cl\mathcal{N}10,l}^{\rm T} \\
	D_{cl\mathcal{N}11,l}^{\rm T} \\
	D_{cl\mathcal{N}12,l}^{\rm T}
	\end{matrix}\right]\left[\begin{matrix}
	C_{cl\mathcal{N}1,l}^{\rm T}\\ D_{cl\mathcal{N}10,l}^{\rm T}\\ D_{cl\mathcal{N}11,l}^{\rm T}\\ D_{cl\mathcal{N}12,l}
	\end{matrix}\right]^{\rm T} \nonumber \\
	&+ \frac{1}{\gamma}\left[\begin{matrix}
	C_{cl2}^{\rm T} \\
	D_{cl20}^{\rm T} \\
	D_{cl21}^{\rm T} \\
	D_{cl22}^{\rm T}
	\end{matrix}\right]\left[\begin{matrix}
	C_{cl2} & D_{cl20} & D_{cl21} & D_{cl22}
	\end{matrix}\right]<0 
	\end{align}
	multiplying vector $[x_{cl}^{\rm T}~p^{\rm T}~w^{\rm T}~d^{\rm T}]^{\rm T}$ from the right side of condition (\ref{eq:20}) and its transpose to the left side, and based on the relationship $p^{\rm T}\Gamma p=\sum_{k=1}^{s}p_k^{\rm T}X_{m_k}p_k+\sum_{k=s+1}^{n_q}\chi_k p_k^{\rm T}p_k$, we obtain
	\begin{align}\label{eq:21}
	&\dot{V} - \sum_{k=1}^{s} p_k^{\rm T}X_{m_k}p_k- \sum_{k=s+1}^{n_q}\chi_k p_k^{\rm T}p_k-\sum_{l=1}^{N_{n_\psi}}\lambda_{l}w^{\rm T}w -\gamma d^{\rm T}d\nonumber \\
	&+\sum_{k=1}^{s}z_{\Delta 1,k}^{\rm T}X_{m_k}z_{\Delta 1,k}+\sum_{k=s+1}^{n_q}\chi_k z_{\Delta 1,k}^{\rm T}z_{\Delta 1,k} \nonumber \\
	&+\sum_{l=1}^{{N_{n_\psi}}}\lambda_lz_{\mathcal{N}1,l}^{\rm T}z_{\mathcal{N}1,l}+\frac{1}{\gamma}
	e^{\rm T}e<0.
	\end{align}
	
	According to the conclusion of $J$-spectral factorization, we have $W_{\Delta, k}=\text{diag}[X_{m_k}, -X_{m_k}]$ for $k\in\textbf{I}[1,s]$ and $W_{\Delta, k}=\text{diag}[\chi_k, -\chi_k]$ for all $k\in\textbf{I}[s+1,n_q]$ and $W_{\mathcal{N},l}=\text{diag}[\lambda_l, -\lambda_l]$ for all $l\in\textbf{I}[1,N_{n_\psi}]$. After simple transformation, the inequality (\ref{eq:18}) is rewritten as:
	\begin{align}\label{eq:22}
	\dot{V}+\sum_{k=1}^{n_q}z_{\Delta, k}^{\rm T}W_{\Delta, k}z_{\Delta, k}+\sum_{l=1}^{N_{n_\psi}}z_{\mathcal{N},l}^{\rm T}W_{\mathcal{N},l}z_{\mathcal{N},l}  <\gamma d^{\rm T}d-\frac{1}{\gamma}e^{\rm T}e.
	\end{align}
	By integrating both sides of the above performance evaluated inequality from $t=0$ to $t=T$ with zero initial conditions, then we have
	\begin{align}\label{eq:23}
	V(x_{cl}(T)) &+ \sum_{k=1}^{n_q}\int_{0}^{T}z_{\Delta, k}^{\rm T}(t)W_{\Delta, k}z_{\Delta, k}(t)dt \nonumber \\
	&+\sum_{l=1}^{N_{n_\psi}}\int_{0}^{T}z_{\mathcal{N},l}^{\rm T}(t)W_{\mathcal{N},l}z_{\mathcal{N},l}(t)dt \nonumber \\
	&< \gamma\int_{0}^{T} d^{\rm T}(t)d(t)dt - \frac{1}{\gamma} \int_{0}^{T}e^{\rm T}(t)e(t)dt.
	\end{align}
	From the definition of IQC and the non-negativity of the Lyapunov function $V$, it can be clearly seen that these three terms on the left of above inequality are all positive, and (\ref{eq:23}) becomes
	\begin{align}\label{eq:24}
	\frac{1}{\gamma} \int_{0}^{T} e^{\rm T}(t)e(t)dt<\gamma\int_{0}^{T}d^{\rm T}(t)d(t)dt.
	\end{align}
	Therefore, the augmented closed-loop system (\ref{eq:16}) with input saturation is stable and performance inequality $\|e\|_2<\gamma\|d\|_2$ holds for all bounded, time-varying structured uncertainties $\Delta\in\bm{\Delta}$ and dead-zone nonlinearity.
	
	Furthermore, taking the inequality (\ref{eq:20}) into consideration and using Schur complement, we have the matrix inequality (\ref{eq:18}) with the existence of both pairs $(X_{m_k},X_{m_k}^{-1})$ for $k\in\textbf{I}[1,s]$ and $(\chi_k,\chi_k^{-1})$ for $k\in\textbf{I}[s+1,n_q]$ as well as $(\lambda_l,\lambda_l^{-1})$ for $l\in\textbf{I}[1,N_{n_\psi}]$.
\end{proof}

Two facts need to be emphasized here. First, we obtain the scaled bounded real lemma within the framework of integral quadratic constraints, which ensures the robust stabilization of closed-loop system (\ref{eq:16}) with bounded, time-varying structured uncertainties and dead-zone nonlinearity, also minimizes the $\mathcal{L}_2$ gain performance $\gamma$. Unlike the proof process previous literature  referred to \cite{CF2016}, namely, introducing a commutable scaling matrix set $\bm{D}$ such that $\Gamma\in\bm{D}$ makes the inequality $q^{\rm T}\Gamma^{1/2}(I-\Delta^{\rm T}\Delta)\Gamma^{1/2}q\ge 0$ hold, while in our work the similar results about scaled bounded real lemma are derived using IQC theory. The main diagonal elements of matrix (\ref{eq:18}) keep consistent, but due to the different realization of closed-loop systems, the explicit form of (\ref{eq:18}) slightly changed. Second, to capture the input-output behavior of each time-varying structured uncertainty $\Delta\in\bm{\Delta}$, we employ a static IQC to realize sBRL (\ref{eq:18}) where scaling matrix is constructed by $\Gamma=\text{diag}\{X_{m_1},\ldots,X_{m_s},\chi_{s+1},\ldots,\chi_{n_q}\}\in\mathbb{S}_{+}^{n_q}$.

The next objective is to derive an $\mathcal{H}_\infty$ controller synthesis condition based on Lemma \ref{Lem:3}. After a congruent transformation and an inequality relaxation, the synthesis condition can be expressed as a set of LMIs, as stated below.
\begin{table*}[htbp]
	\begin{eqnarray}\label{eq:25}
	\left[\begin{matrix}
	He\left\{\left[\begin{matrix}
	A & B_0 & 0\\
	0 & -\alpha I_{n_u} & 0 \\
	0 & 0 & A_{\mathcal{N}\bar{\Psi}}
	\end{matrix}\right]Q+\left[\begin{matrix}
	0 \\
	I_{n_u}\\
	B_{\bar{\Psi}1}
	\end{matrix}\right]\hat{F}_c\right\} & \star & \star & \star & \star & \star & \star\\
	\hat{\Gamma}^{\rm T}\left[\begin{matrix}
	B_1^{T} & 0 & 0
	\end{matrix}\right] & -\hat{\Gamma} & \star & \star & \star & \star & \star\\
	\hat{\lambda}\left[\begin{matrix}
	-B_0^{\rm T} & 0 & B_{\bar{\Psi}2}^{\rm T}
	\end{matrix}\right]+\hat{H}_c^{\rm T}\left[\begin{matrix}
	0 & I & B_{\bar{\Psi}1}^{\rm T}
	\end{matrix}\right] & 0 & \sum_{l=1}^{N_{n_\psi}}(\hat{\lambda}_l-2\hat{\lambda})I_{n_u} & \star & \star & \star & \star\\
	\left[\begin{matrix}
	B_2^{\rm T} & 0 & 0
	\end{matrix}\right] & 0 & 0 & -\gamma I_{n_d} & \star & \star & \star\\
	\Omega_{51} & \Omega_{52} & \Omega_{53} & \Omega_{54} & -\hat{\Gamma} & \star & \star\\
	\Omega_{61} & 0 & \Omega_{63} & 0 & 0 & -\Lambda
	& \star\\
	\left[\begin{matrix}
	C_1 & D_{10} & 0
	\end{matrix}\right]Q & D_{11}\hat{\Gamma} & -D_{10}\hat{\lambda} & D_{12} & 0 & 0 & -\gamma I_{n_e}
	\end{matrix}\right]<0
	%He\left\{\left[\begin{array}{cccc}
	%          A & B_0 & 0 & 0\\
	%          0 & -\alpha & 0 & 0 \\
	%          0 & 0 & A_{\mathcal{N}\Psi} & 0 \\
	%          B_{\Delta\Psi1}C_0 & B_{\Delta\Psi1}D_{00} & 0 & A_{\Delta\Psi}
	%        \end{array}\right]Q+\left[\begin{array}{c}
	%          0 \\
	%          1 \\
	%          B_{\Psi1} \\
	%          0
	%        \end{array}\right]\hat{F}_c\right\}+2\sigma Q<0 \nonumber \\
	%        \left[\begin{array}{cc}
	%          \sin{\theta}He\left\{\left[\begin{array}{cccc}
	%          A & B_0 & 0 & 0\\
	%          0 & -\alpha & 0 & 0 \\
	%          0 & 0 & A_{\mathcal{N}\Psi} & 0 \\
	%          B_{\Delta\Psi1}C_0 & B_{\Delta\Psi1}D_{00} & 0 & A_{\Delta\Psi}
	%        \end{array}\right]Q+\left[\begin{array}{c}
	%          0 \\
	%          1 \\
	%          B_{\Psi1} \\
	%          0
	%        \end{array}\right]\hat{F}_c\right\} & * \\
	%        \cos{\theta}\left[\right] & \sin{\theta}He\left\{\left[\begin{array}{cccc}
	%          A & B_0 & 0 & 0\\
	%          0 & -\alpha & 0 & 0 \\
	%          0 & 0 & A_{\mathcal{N}\Psi} & 0 \\
	%          B_{\Delta\Psi1}C_0 & B_{\Delta\Psi1}D_{00} & 0 & A_{\Delta\Psi}
	%        \end{array}\right]Q+\left[\begin{array}{c}
	%          0 \\
	%          1 \\
	%          B_{\Psi1} \\
	%          0
	%        \end{array}\right]\hat{F}_c\right\}
	%        \end{array}\right]
	\end{eqnarray}
\end{table*}

\begin{Thm}\label{Thm:1}
	Consider the uncertain LFT systems with input saturation (\ref{eq:1}). If there exist positive-definite matrices $Q=P^{-1}\in\mathbb{S}_+^{n_x+n_u+n_\psi}$,
	$\hat\Gamma=\text{diag}[X_{m_1}^{-1},\ldots,X_{m_s}^{-1},\chi_{s+1}^{-1},\ldots,\chi_{n_q}^{-1}]\in\mathbb{S}_+^{n_q}$, positive scaling factors $\hat{\lambda}_{l}=\lambda_{l}^{-1}\in\mathbb{R}_+$ for all $l\in\textbf{I}[1,N_{n_\psi}]$, $\hat{\lambda}\in\mathbb{R}_+$, $\Lambda = \text{diag}\left\{\lambda_1^{-1},\ldots,\lambda_{N_{n_\psi}}^{-1}\right\}$, rectangular matrices
	$\hat{F}_c=F_cQ\in\mathbb{R}^{n_u\times(n_x+n_u+n_\psi)}$, $\hat{H}_c=H_c\hat{\lambda}\in\mathbb{R}^{n_u\times n_u}$ and a positive scalar $\gamma\in\mathbb{R}_+$ such that LMI condition (\ref{eq:25}) holds with
	\begin{align} \label{eq:26}
	\Omega_{51} &= \left[\begin{array}{c}
	\left[\begin{matrix}
	D_{\Delta\bar{\Psi}1,1}C_0 & D_{\Delta\bar{\Psi}1,1}D_{00} & 0
	\end{matrix}\right]Q \\
	\vdots \\
	\left[\begin{matrix}
	D_{\Delta\bar{\Psi}1,n_q}C_0 & D_{\Delta\bar{\Psi}1,n_q}D_{00} & 0
	\end{matrix}\right]Q
	\end{array}\right] \nonumber \\
	\Omega_{52} &= \left[\begin{array}{c}
	\left(
	D_{\Delta\bar{\Psi}1,1}D_{01}+D_{\Delta\bar{\Psi}2,1}
	\right)\hat{\Gamma} \\
	\vdots \\
	\left(
	D_{\Delta\bar{\Psi}1,n_q}D_{01}+D_{\Delta\bar{\Psi}2,n_q}
	\right)\hat{\Gamma}
	\end{array}\right] \nonumber \\
	\Omega_{53} &= \left[\begin{matrix}
	-D_{\Delta\bar{\Psi}1,1}D_{00}\hat{\lambda} \\
	\vdots \\
	-D_{\Delta\bar{\Psi}1,n_q}D_{00}\hat{\lambda}
	\end{matrix}\right] \nonumber \\
	\Omega_{54} &= \left[\begin{matrix}
	D_{\Delta\bar{\Psi}1,1}D_{02} \\
	\vdots \\
	D_{\Delta\bar{\Psi}1,n_q}D_{02}
	\end{matrix}\right], \nonumber \\
	\Omega_{61} &= \left[\begin{array}{c}
	\left[\begin{matrix}
	0 & 0 & C_{\mathcal{N}\bar{\Psi},1}
	\end{matrix}\right]Q+D_{\mathcal{N}\bar{\Psi}1,1}\hat{F}_c \\
	\vdots  \\
	\left[\begin{matrix}
	0 & 0 & C_{\mathcal{N}\bar{\Psi},{N_{n_\psi}}}
	\end{matrix}\right]Q+D_{\mathcal{N}\bar{\Psi}1,{N_{n_\psi}}}\hat{F}_c
	\end{array}\right] \nonumber \\
	\Omega_{63} &= \left[\begin{array}{c}
	D_{\mathcal{N}\bar{\Psi}1,1}\hat{H}_c+D_{\mathcal{N}\bar{\Psi}2,1}\hat{\lambda} \\
	\vdots \\
	D_{\mathcal{N}\bar{\Psi}1,{N_{n_\psi}}}\hat{H}_c+D_{\mathcal{N}\bar{\Psi}2,{N_{n_\psi}}}\hat{\lambda}
	\end{array}\right].
	\end{align}
	Then the uncertain LFT system (\ref{eq:1}) with time-varying parameters and input saturation is stabilized by the state-feedback $\mathcal{H}_\infty$ controller (\ref{eq:15}) and the performance level of its $\mathcal{L}_2$
	gain is no larger than $\gamma$ for all uncertainties $\Delta\in\bm\Delta$ and nonlinearity.
\end{Thm}
\begin{proof}
	From Lemma \ref{Lem:3}, it can be seen that (\ref{eq:18}) is a non-convex condition due to the sum of scaling factors, i.e., $\sum_{l=1}^{N_{n_\psi}}\lambda_l$. To tackle this problem, we define $\hat{\lambda}_l=\lambda_l^{-1}$ for all $l\in\textbf{I}[1,N_{n_\psi}]$, and according to the inequality relationship that
	$-\hat{\lambda}_l^{-1}\le \hat{\lambda}^{-2}\hat{\lambda}_l-2\hat{\lambda}^{-1}$, the inequality (\ref{eq:18}) becomes
	\begin{align}\label{eq:27}
	\left[\begin{matrix}
	He\left\{PA_{cl}\right\} & \star & \star  \\
	B_{cl0}^{\rm T}P & -\Gamma & \star  \\
	B_{cl1}^{\rm T}P & 0 & -\sum_{l=1}^{N_{n_\psi}}(\hat{\lambda}^{-2}\hat{\lambda}_l-2\hat{\lambda}^{-1})I_{n_u}  \\
	B_{cl2}^{\rm T}P & 0 & 0   \\
	\Theta_{cl\Delta 1} & \Theta_{cl\Delta 10} & \Theta_{cl\Delta 11}\\
	\Xi_{cl\mathcal{N}1} & \Xi_{cl\mathcal{N}10} & \Xi_{cl\mathcal{N}11}  \\
	C_{cl2} & D_{cl20} & D_{cl21}
	\end{matrix}\right. \nonumber \\
	\left.\begin{matrix}
	\star & \star & \star & \star \\
	\star & \star & \star & \star \\
	\star & \star & \star & \star \\
	-\gamma I_{n_d}  & \star & \star & \star \\
	\Theta_{cl\Delta12} & -\Gamma^{-1} & \star & \star \\
	\Xi_{cl\mathcal{N}12} & 0 & -\Lambda & \star \\
	D_{cl22} & 0 & 0 & -\gamma I_{n_e}
	\end{matrix}\right]<0.
	\end{align}
	Then, by multiplying diagonal matrix $\text{diag}\{P^{-1},\Gamma^{-1},\hat\lambda,I_{n_d},\notag\\I_{n_{\Delta z},1},\ldots,I_{n_{\Delta z},n_q},I_{n_{\mathcal{N}z},1},\ldots,I_{n_{\mathcal{N}z},N_{n_\psi}},I_{n_e}\}$ to the right and its transpose from the left on both sides of inequality (\ref{eq:27}), we obtain the following results
	\begin{align}\label{eq:28}
	\left[\begin{matrix}
	P^{-{\rm T}}A_{cl}^{\rm T}+A_{cl}P^{-1} & \star & \star  \\
	\Gamma^{-1}B_{cl0}^{\rm T} & -\Gamma^{-1} & \star  \\
	\hat{\lambda}B_{cl1}^{\rm T} & 0 & -\sum_{l=1}^{N_{n_\psi}}(\hat{\lambda}_l-2\hat{\lambda})I_{n_u}  \\
	B_{cl2}^{\rm T} & 0 & 0   \\
	\Theta_{cl\Delta 1}P^{-1} & \Theta_{cl\Delta 10}\Gamma^{-1} & \Theta_{cl\Delta 11}\hat{\lambda}\\
	\Xi_{cl\mathcal{N}1}P^{-1} & \Xi_{cl\mathcal{N}10}\Gamma^{-1} & \Xi_{cl\mathcal{N}11}\hat{\lambda}  \\
	C_{cl2}P^{-1} & D_{cl20}\Gamma^{-1} & D_{cl21}\hat{\lambda}
	\end{matrix}\right. \nonumber \\
	\left.\begin{matrix}
	\star & \star & \star & \star \\
	\star & \star & \star & \star \\
	\star & \star & \star & \star \\
	-\gamma I_{n_d} & \star & \star & \star \\
	\Theta_{cl\Delta12} & -\Gamma^{-1} & \star & \star \\
	\Xi_{cl\mathcal{N}12} & 0 & -\Lambda & \star \\
	D_{cl22} & 0 & 0 & -\gamma I_{n_e}
	\end{matrix}\right]<0. \nonumber \\
	&\hspace*{-1.5ex} \hspace*{-1.5ex}&
	\end{align}
	Specifically, denoting $Q=P^{-1}$, $\hat{\Gamma}= \Gamma^{-1}$, $\hat{F}_c=F_cQ$, and $\hat{H}_c=\hat{\lambda}H_c$, and substituting the associated closed-loop system matrices, we obtain condition (\ref{eq:25}) eventually.
\end{proof}

The synthesis condition \eqref{eq:25} is given in terms of LMIs and can therefore be solved efficiently by standard interior-point algorithms. The feasibility problem associated with \eqref{eq:25} can be cast as the following convex optimization problem:
	\begin{align}\label{eq:29}
	\begin{array}{cc}
	\min\limits_{Q,D_{m_{k_1}},\hat{\chi}_{k_2},\hat{\lambda}_{l},\hat{F}_c,\hat{H}_c, \forall{k_1}\in\textbf{I}[1,s],\forall{k_2}\in\textbf{I}[s+1,n_q],\forall{l}\in\textbf{I}[1,N_{n_\psi}]} & \gamma \\
	\text{s.t.} \ \eqref{eq:25}.
	\end{array}
	\end{align}

Note that the inequality relaxation
\[
-\hat{\lambda}_l^{-1}\le \hat{\lambda}^{-2}\hat{\lambda}_l-2\hat{\lambda}^{-1}, \qquad \forall l\in\mathbf{I}[1,N_{n_\psi}],
\]
is introduced when converting the synthesis condition into LMIs. As pointed out in \cite{CF2016,CF2017}, the LMI condition \eqref{eq:25} is therefore only sufficient for robust $\mathcal{H}_\infty$ controller synthesis, and some conservatism may be introduced.

\section{Numerical Examples}

\subsection{A Second-Order LFT Uncertain System}

In this subsection, a second-order uncertain LFT system is used to illustrate the effectiveness of the proposed robust $\mathcal{H}_\infty$ state-feedback controller. Let us consider a second-order uncertain LFT system with input saturation
\begin{align}\label{eq:32}
\dot{x} &= \left[\begin{matrix}
0 & 1 \\
-10 & -8
\end{matrix}\right]x + \left[\begin{matrix}
0 \\
1
\end{matrix}\right]p + \left[\begin{matrix}
0 \\
-1
\end{matrix}\right]d + \left[\begin{matrix}
0 \\
0.1
\end{matrix}\right]Sat(u) \nonumber \\
q &= \left[\begin{matrix}
2 & -1
\end{matrix}\right]x + d + 0.3Sat(u)\nonumber \\
e &= \left[\begin{matrix}
-1 & 1
\end{matrix}\right]x  + p+ 0.5d+0.1Sat(u)\nonumber \\
p   &= \Delta q
\end{align}
where $Sat(u)$ denotes the input saturation nonlinearity. By introducing the auxiliary dynamics $\dot{u}=-\alpha u+v$ and the dead-zone function $w=u-Sat(u)$, the system \eqref{eq:32} can be rewritten as
\begin{align}\label{eq:33}
\left[\begin{matrix}
\dot{x} \\
\dot{u} \\
q \\
e
\end{matrix}\right]&=\left[\begin{array}{ccc:cccc}
0 & 1 & 0 & 0 & 0 & 0 & 0 \\
-10 & -8 & 0.1 & 1 & -0.1 & -1 & 0 \\
0 & 0 & -\alpha & 0 & 0 & 0 & 1\\
\hdashline
2 & -1 & 0.3 & 0 & -0.3 & 1 & 0\\
-1 & 1 & 0.1 & 1 & -0.1 & 0.5 & 0
\end{array}\right]\left[\begin{matrix}
x \\
u \\
p \\
w \\
d \\
v
\end{matrix}\right] \nonumber \\
p &= \Delta q \nonumber \\
w &= \mathcal{N}(u) = u-Sat(u)
\end{align}
where
\begin{align}\label{eq:34}
A &= \left[\begin{matrix}
0 & 1 & 0 \\
-10 & -8 & 0.1 \\
0 & 0 & -\alpha
\end{matrix}\right], B_0 = \left[\begin{matrix}
0 \\
0.1 \\
0
\end{matrix}\right], B_1 = \left[\begin{matrix}
0\\
1 \\
0
\end{matrix}\right] \nonumber \\
B_2 &= \left[\begin{matrix}
0 \\
-1 \\
0
\end{matrix}\right],C_0 = \left[\begin{matrix}
2 & -1 & 0.3 \\
-1 & 1 & 0.1
\end{matrix}\right],\nonumber \\
D_{00}&=0.3,D_{01}=0,D_{02}=1,\nonumber \\
D_{10}&=0.1,D_{11}=1,D_{12}=0.5.
\end{align}
As discussed previously, the Popov IQC \eqref{eq:7}, the Zames--Falb IQC \eqref{eq:9}, and the sector-bound IQC \eqref{eq:8}, together with the scaling factors $\lambda_1,\lambda_2,\lambda_3$, are used to characterize the dead-zone nonlinearity $\mathcal{N}(u)$, namely,
$\Pi_{\mathcal{N}}=\lambda_1\Pi_{\bar{P}}+\lambda_2\Pi_{\bar{ZF}}+\lambda_3\Pi_{\bar{S}}$. After $J$-spectral factorization, the associated IQC-induced dynamics are described in the state-space form
\begin{eqnarray}
A_{\mathcal{N}\bar{\Psi}}=\left[\begin{array}{ccc}
A_{\mathcal{N}\bar{\Psi}_{\bar{P}}} & 0 & 0 \\
0 & A_{\mathcal{N}\bar{\Psi}_{\bar{ZF}}} & 0 \\
0 & 0 & A_{\mathcal{N}\bar{\Psi}_{\bar{S}}}
\end{array}\right], \nonumber \\
B_{\mathcal{N}\bar{\Psi}1} = \left[\begin{array}{c}
B_{\mathcal{N}\bar{\Psi}_{\bar{P}}1} \\
B_{\mathcal{N}\bar{\Psi}_{\bar{ZF}}1} \\
B_{\mathcal{N}\bar{\Psi}_{\bar{S}}1}
\end{array}\right], B_{\mathcal{N}\Psi2} = \left[\begin{array}{c}
B_{\mathcal{N}\bar{\Psi}_{\bar{P}}2} \\
B_{\mathcal{N}\bar{\Psi}_{\bar{ZF}}2} \\
B_{\mathcal{N}\bar{\Psi}_{\bar{S}}2}
\end{array}\right], \nonumber \\
C_{\mathcal{N}\bar{\Psi}}=\left[\begin{array}{ccc}
C_{\mathcal{N}\bar{\Psi}_{\bar{P}}} & 0 & 0 \\
0 & C_{\mathcal{N}\bar{\Psi}_{\bar{ZF}}} & 0 \\
0 & 0 & C_{\mathcal{N}\bar{\Psi}_{\bar{S}}}
\end{array}\right],  \nonumber \\
D_{\mathcal{N}\bar{\Psi}1} = \left[\begin{array}{c}
D_{\mathcal{N}\bar{\Psi}_{\bar{P}}1} \\
D_{\mathcal{N}\bar{\Psi}_{\bar{ZF}}1} \\
D_{\mathcal{N}\bar{\Psi}_{\bar{S}}1}
\end{array}\right], D_{\mathcal{N}\bar{\Psi}2} = \left[\begin{array}{c}
D_{\mathcal{N}\bar{\Psi}_{\bar{P}}2} \\
D_{\mathcal{N}\bar{\Psi}_{\bar{ZF}}2} \\
D_{\mathcal{N}\bar{\Psi}_{\bar{S}}2}
\end{array}\right].
\end{eqnarray}
Detailed $J$-spectral factorization is given in the Appendix A. On the other hand, we take $n_q=1$ for simplicity, and the uncertainty $\Delta$ is bounded by $1$, i.e., $|\Delta|\le1$. The IQC is chosen by (\ref{eq:11}) with $b=1$. $\Psi_{\Delta}=\left[\begin{array}{cc}
b & 0 \\
0 & 1
\end{array}\right]$ also can be expressed in the static state-space form. Note that $\Psi_{\Delta}$ is a constant matrix, then after state-space realization, $A_{\Delta\bar{\Psi}}$, $B_{\Delta\bar{\Psi}1}$, $B_{\Delta\bar{\Psi}2}$ and $C_{\Delta\bar{\Psi}}$ are empty matrices, and only constant matrices $D_{\Delta\bar{\Psi}1}=b$ and $D_{\Delta\bar{\Psi}2}=0$ need to be substituted into the synthesis condition (\ref{eq:25}) when solving this LMI. Using the proposed $\mathcal{H}_\infty$ state-feedback controller, we next examine the resulting $\mathcal{L}_2$-gain level from the disturbance input $d$ to the performance output $e$ from two perspectives.
\begin{table}[!t]
	\caption{Performance level $\gamma$ with regard to different IQCs-based control strategies as $\alpha$ increases}
	\label{Tab:1}
	\centering
	\begin{tabular}{c|cccccc}
		\hline\hline
		$\alpha$  & $2$ & $5$ & $7$ & $10$ & $15$ & $20$\\
		$\gamma_P$ & 3.041 & 3.041 & 3.041 & 3.040 & 3.041 & 3.040 \\
		$\gamma_{ZF}$ & 1.520 & 1.520 & 1.520 & 1.520 & 1.520 & 1.520 \\
		$\gamma_{S}$ & 8.155 & 8.153 & 8.154 & 8.155 & 8.154 & 8.156 \\
		$\gamma_{M}$  & 1.508 & 1.508 & 1.508 & 1.508 & 1.508 & 1.508 \\
		\hline
		$\alpha$ & $30$ & $40$ & $50$ & $60$ & 70 & $100$\\
		$\gamma_P$ & 3.040 & 3.041 & 3.041 & 3.040 & 3.041 & 3.040 \\
		$\gamma_{ZF}$ & 1.524 & 1.528 & 1.530 & 1.531 & 1.531 & 1.532 \\
		$\gamma_S$ & 8.155 & 8.155 & 8.155 & 8.154 & 8.155 & 8.155 \\
		$\gamma_{M}$ & 1.508 & 1.508 & 1.508 & 1.508 & 1.508 & 1.508 \\
		\hline
		\hline
	\end{tabular}
\end{table}
First, we investigate the influence of the inertial element $\frac{1}{s+\alpha}$ introduced in the dead-zone loop on the achievable $\mathcal{L}_2$-gain performance. As shown in Table \ref{Tab:1} and Fig. \ref{Fig:2}, the $\mathcal{L}_2$ gain $\gamma$ will maintain around 1.5086 when $\alpha$ increases from $2$ to $100$. Therefore, although the inertial term is introduced to render the Popov IQC proper, the additional pole associated with $\alpha$ has little influence on the achievable disturbance attenuation level.
\begin{figure}
	\centering
	\includegraphics[width=0.45\textwidth]{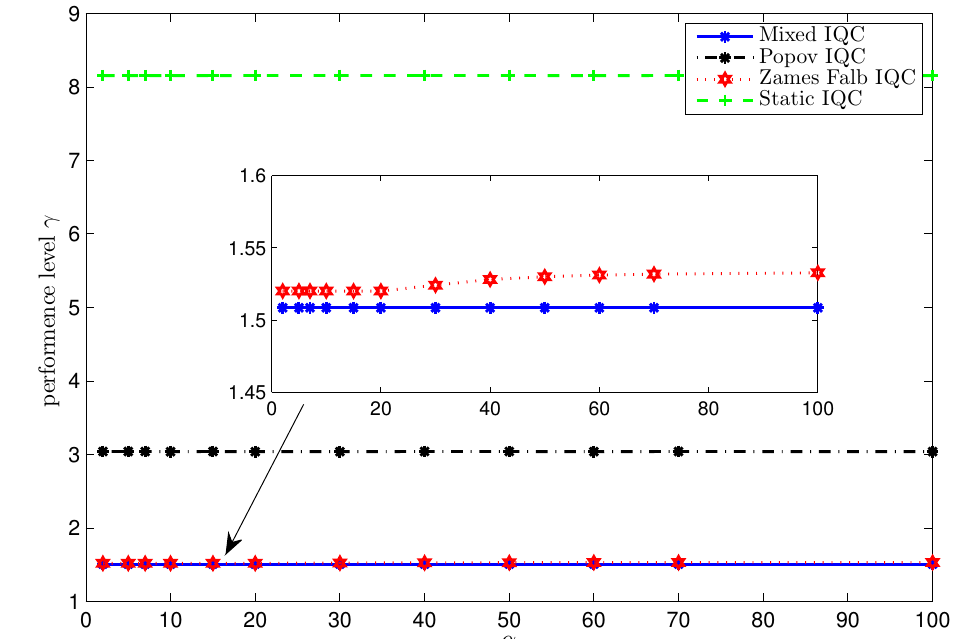}
	\vspace{-0.5cm}
	\caption{Performance level $\gamma$ using different IQC's strategies as $\alpha$ increases. \label{Fig:2}}
\end{figure}

Second, based on Lemma \ref{Lem:3}, we examine the improvement in robust $\mathcal{H}_\infty$ performance achieved by the proposed mixed-IQC-based state-feedback controller, and compare it with controller designs based on a single IQC and the standard sBRL in \cite{PP1995}. The value of $\alpha$ is initiated as $1$, and Popov, Zames-Falb and static IQC multipiers are picked individually. Other initial conditions do not change, then the $\mathcal{H}_\infty$ performance $\gamma$ with different control strategies are also exhibited in Table \ref{Tab:1}. In particular, the mixed-IQC-based $\gamma$-suboptimal $\mathcal{H}_\infty$ state-feedback controller \eqref{eq:15} yields a smaller performance level $\gamma$ than the corresponding single-IQC-based designs. Compared with single IQC, we have obtained the desired $\mathcal{L}_2$ gain using the mixed IQCs with scaling factors. This remarkable characteristic attributes to the introduction of scaling factors $\lambda_l$ for $l\in\textbf{I}[1,N_{n_\psi}]$ ($k=1$ and $l=3$ in this example). By regulating the weight of each IQC and solving the convex synthesis condition (\ref{eq:25}), we are able to obtain the optimised performance level $\gamma$ all the time. These results show that the proposed controller design strategy improves the closed-loop behavior and provides better attenuation of the external disturbance $d$.

%\begin{table}[!t]
%\caption{Performance level $\gamma$ with regard to different IQCs-based control strategies}
%\label{Tab:2}
%\centering
%\begin{tabular}{c|cccc}
%\hline\hline
%  & Popov & ZF & Static  & Mixed\\
%\hline
%$\gamma$  & 4.1407 & 4.2636 & 51.0084 & 4.0526 \\
%\hline
%\hline
%\end{tabular}
%\end{table}

%\begin{figure}
%\centering
%\includegraphics[width=0.5\textwidth]{Per_gamma.eps}
%%\vspace{-1cm}
%\caption{Performance level $\gamma$ with regard to different IQCs strategies. \label{Fig:3}}
%\end{figure}

Moreover, we usually obtain a pole of closed-loop system with large magnitude in the left half of complex plane after solving synthesis condition (\ref{eq:25}). As \cite{MP1996} proposed, one can place the poles in an LMI region $\mathcal{S}(\rho,\theta)$ of complex numbers $\xi+j\eta$ such that $\xi<-\rho<0$ and $\tan{\theta}\xi<-|\eta|$. The corresponding LMI conditions are described as:
\begin{align}
&A_{cl}X_D+X_DA_{cl}^{\rm T}+2\rho X_D<0 \label{eq:30} \\
&\left[\begin{matrix}\label{eq:31}
\sin{\theta}(A_{cl}X_D+X_DA_{cl}^{\rm T}) & \cos{\theta}(A_{cl}X_D-X_DA_{cl}^{\rm T}) \\
\cos{\theta}(A_{cl}X_D-X_DA_{cl}^{\rm T}) & \sin{\theta}(A_{cl}X_D+X_DA_{cl}^{\rm T})
\end{matrix}\right]<0.\nonumber \\
&\hspace*{-1.5ex} \hspace*{-1.5ex}&
\end{align}
Here, convexity is enforced by seeking a common solution $X_D=Q>0$ to (\ref{eq:25}), (\ref{eq:30}) and (\ref{eq:31}), and some conservatism could be introduced this way. Though the $\mathcal{H}_\infty$ performance in the terms of disturbance attenuation could be sacrificed, satisfactory transient response of linear systems can be ensured in the simulation.

To this end, assuming the external disturbance $d=0.5\sin(0.5t+\frac{\pi}{3})$ is activated from $t=10s$ to $20s$ and letting parameters be $\alpha=2$ and saturated upper bound be $\bar{u}=0.0003$, we apply the proposed mixed IQCs-based state-feedback controller to this second-order system. Moreover, LMIs (\ref{eq:30}) and (\ref{eq:31}) are incorporated to controller synthesis condition (\ref{eq:25}), and parameters $\rho$ and $\theta$ in the LMI region $\mathcal{S}(\rho,\theta)$ are chosen as $1.0$ and $\frac{\pi}{3}$, respectively. After synthesis, controller gains are calculated as $ F_{c} = \left[
0.0109 ~ 0.0041 ~ 0.0589 ~ 0.3229 ~ -0.0754 \right.$ $\left.
-0.1386~ -0.0488 ~ -0.0101\right]$ and $ H_c = 0.6536 $.
Scaling factors are $\lambda_1=60.9916$, $\lambda_2=77.1289$ and $\lambda_3=43.7049$. Simulation results are shown in Figs. \ref{Fig:3}-\ref{Fig:4}. Fig. \ref{Fig:3} shows the trajectories of system state, and it can be observed that state variables can reach the equilibrium point after $t=3s$. The poles of closed-loop systems are listed as $-6.4496$, $-1.5491$, $-1.6202$, $-2.3275$, $-2.0000$, $-2.0000$, $-2.0000$ and $-2.3597$. Fig. \ref{Fig:4} shows the saturation input curve, and input enters saturation zone when time $t$ goes through $1-2s$ and $15-16s$. Therefore, we conclude that with the implementation of the proposed $\mathcal{H}_\infty$ controller, even though external disturbance occurs, the closed-loop system maintains robust stability.
\begin{figure}[htbp]
	\centering
	\includegraphics[width=0.45\textwidth]{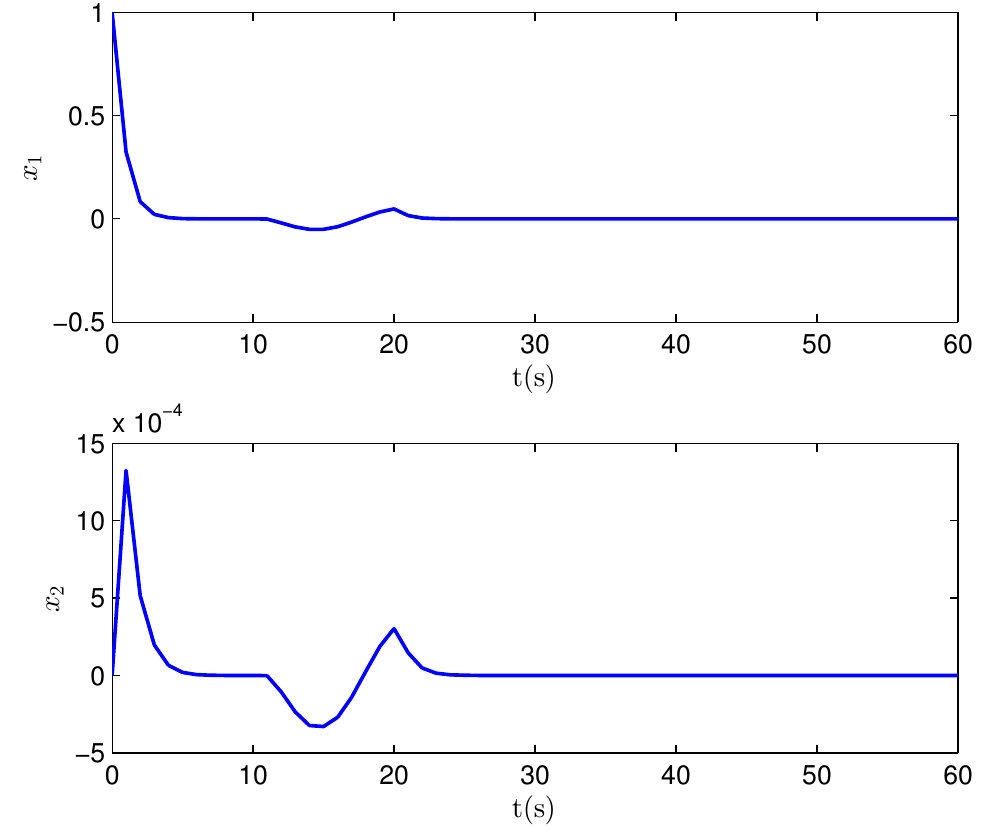}
	%\vspace{-1cm}
	\caption{State trajectory with the implementation of mixed IQC-based $\mathcal{H}_\infty$ controller. \label{Fig:3}}
\end{figure}
\begin{figure}[htbp]
	\centering
	\includegraphics[width=0.45\textwidth]{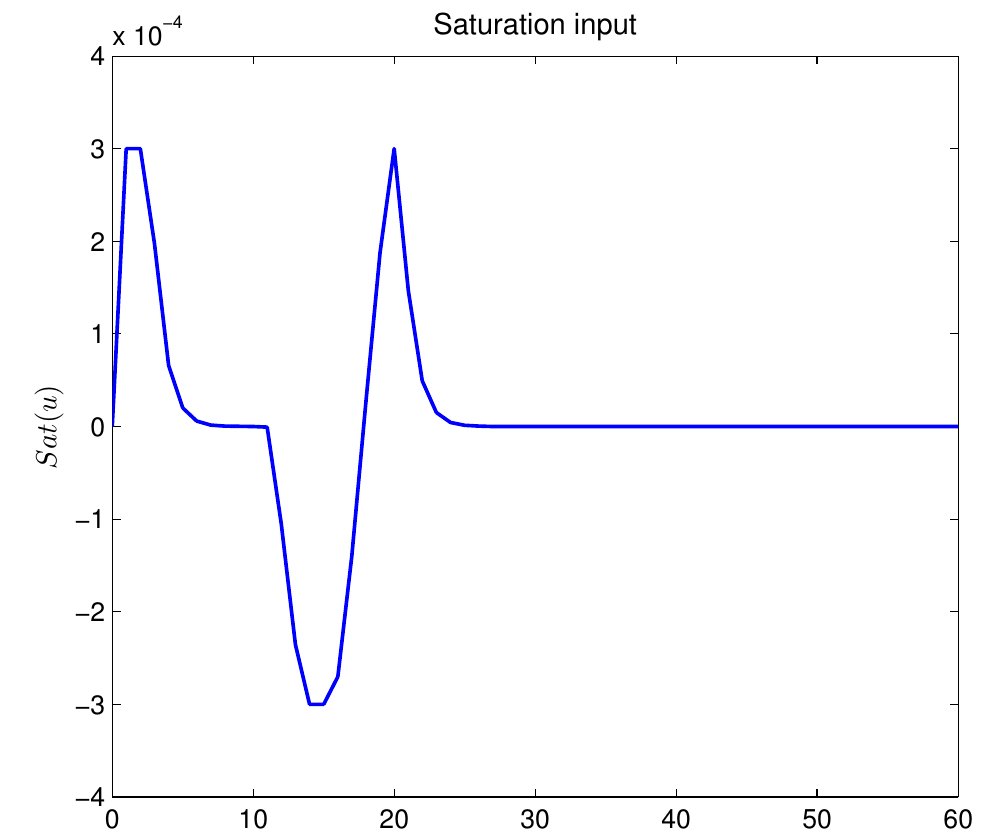}
	%\vspace{-1cm}
	\caption{Input saturation with the implementation of mixed IQC-based $\mathcal{H}_\infty$ controller. \label{Fig:4}}
\end{figure}

\subsection{Cart-spring pendulum system}
\begin{figure}[htbp]
	\centering
	\includegraphics[width=0.4\textwidth]{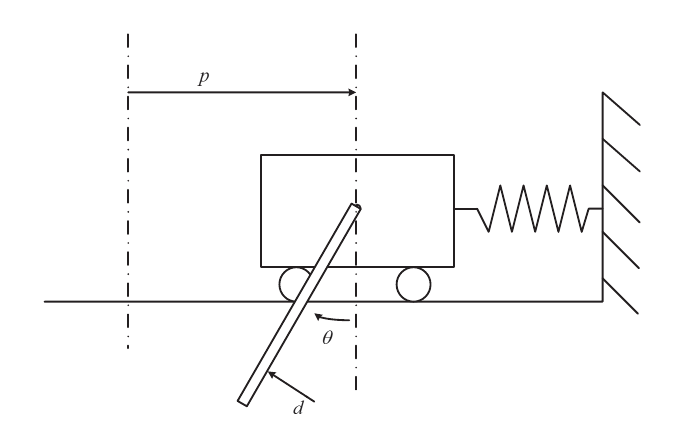}
	%\vspace{-1cm}
	\caption{Damped mass-spring-pendulum system. \label{Fig:5}}
\end{figure}
To further demonstrate the advantage of the dynamic-IQC-based $\mathcal{H}_\infty$ controller over the static sector-bound approach, we next consider a cart--spring--pendulum system \cite{GJIA2003}. From Fig. \ref{Fig:5}, $p$ is the position of this cart, $\theta$ is the angular position of the pendulum, $u$ is the voltage applied to the motor and $d$ is a disturbance force. Let system be $x_1=p$,$x_2=\dot{p}$, $x_3=\theta$ and $x_4=\dot{\theta}$. After linearization around the origin, its state-space equation is given as:
\begin{eqnarray}\label{eq:36}
\dot{x}&\hspace*{-1.5ex}=\hspace*{-1.5ex}& A
x +B_{12}d +B_{10}Sat(u)
\end{eqnarray}
with
\begin{align}
A &= \left[\begin{matrix}
0 & 1 & 0 & 0 \\
-330.46 & -12.15 & -2.44 & 0 \\
0 & 0 & 0 & 1\\
-812.61 & -29.87 & -30.1 & 0
\end{matrix}\right] \nonumber \\
B_0 &= \left[\begin{matrix}
0 \\
2.71762 \\
0 \\
6.68268
\end{matrix}\right], B_2 = \left[\begin{matrix}
0 \\
0 \\
0 \\
15.61
\end{matrix}\right].
\end{align}
We choose the matrices related to the performance output $e$ as follows: $C_{1}=\left[\begin{matrix}
0 & 0 & 1 & 0
\end{matrix}\right]$, $D_{10}=0$ and $D_{12}=0$.

We then analyze the disturbance attenuation performance from the external disturbance to the error output under the proposed dynamic-IQC-based $\mathcal{H}_\infty$ controller. Note that the modified Zames-Falb IQC in (\ref{eq:9}) is selected while solving LMI (\ref{eq:25}) and making simulation, especially this IQC with scaling factor is considered, respectively. We assume that the value of $\alpha$ is set to $1$. When the dynamic IQC with a scaling factor is used in \eqref{eq:25} to solve the optimization problem based on Theorem \ref{Thm:1}, the resulting $\mathcal{L}_2$-gain level is $\gamma_{dyn}=3.022$.

For comparison, the synthesis inequality with the implementation of anti-windup control approach is given by $\dot{V}+\frac{1}{\gamma}e^{\rm T}e-\gamma d^{\rm T}d+\omega^{\rm T}\Lambda(u-\omega)<0$ where $\omega=\mathcal{N}(u)$ is a dead-zone function and $\Lambda$ is a diagonal positive-definite matrix. Through simple transformation by substituting the state-feedback controller $u=F_cx+H_c\omega$, and its controller synthesis condition in the LMI form is written as
\begin{align*}
	\left[\begin{matrix}
		He\left\{A_pQ+B_0\hat{F}_c\right\} & \star \\
		-\Gamma^{\rm T}B_0^{\rm T}+\hat{H}_c^{\rm T}B_0^{\rm T}+\hat{F}_c & \hat{H}_c+\hat{H}_c^{\rm T}-2\Gamma \\
		B_2^{\rm T} & 0 \\
		C_1Q+D_{10}\hat{F}_c & -D_{10}\Gamma+D_{10}\hat{H}_c \\
	\end{matrix}\right. \nonumber \\
	\left.\begin{matrix}
		 \\
		\star & \star \\
		-\gamma I & \star \\
		D_{12} & -\gamma I
	\end{matrix}
	\right]<0
\end{align*}
where $\Gamma = \Lambda^{-1}$, $\hat{F}_c= F_cQ$ and $\hat{H}_c=H_c\Gamma$. By solving the above LMI under the same system parameters, we obtain the suboptimal performance level $\gamma_{sc}=181.142$, which is much larger than $\gamma_{dyn}$. Therefore, the disturbance attenuation performance of the input-saturated LTI system is significantly improved by the proposed dynamic-IQC-based robust $\mathcal{H}_\infty$ controller.

\begin{figure}[htbp]
	\centering
	\includegraphics[width=0.45\textwidth]{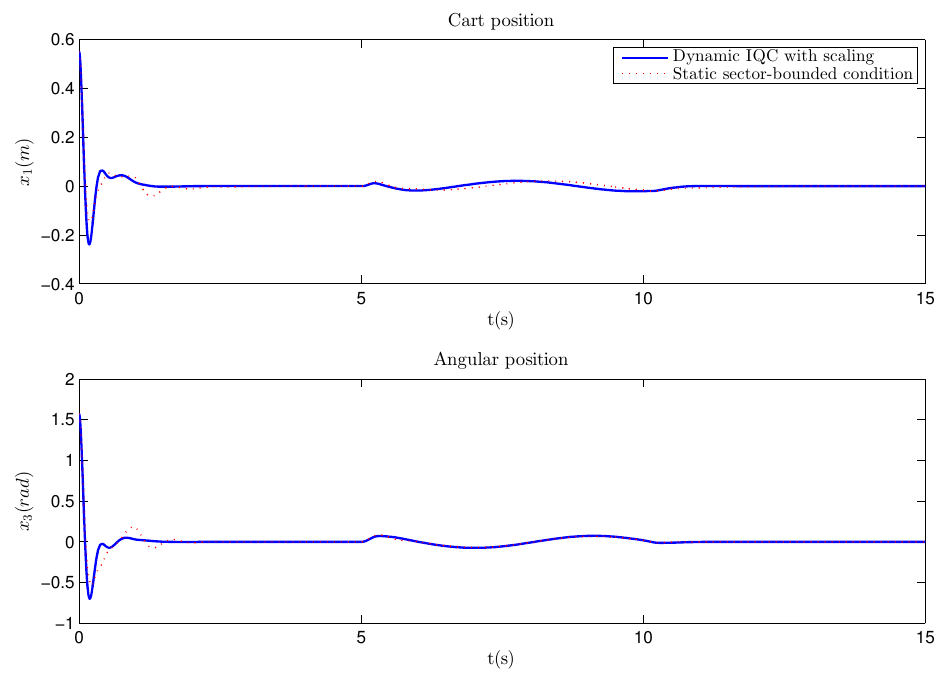}
	%\vspace{-1cm}
	\caption{Trajectories of the cart-spring pendulum. Dynamic IQC with scaling factor (bold solid line); static sector-bound condition (thin dotted line).  \label{Fig:6}}
\end{figure}
\begin{figure}[htbp]
	\centering
	\includegraphics[width=0.45\textwidth]{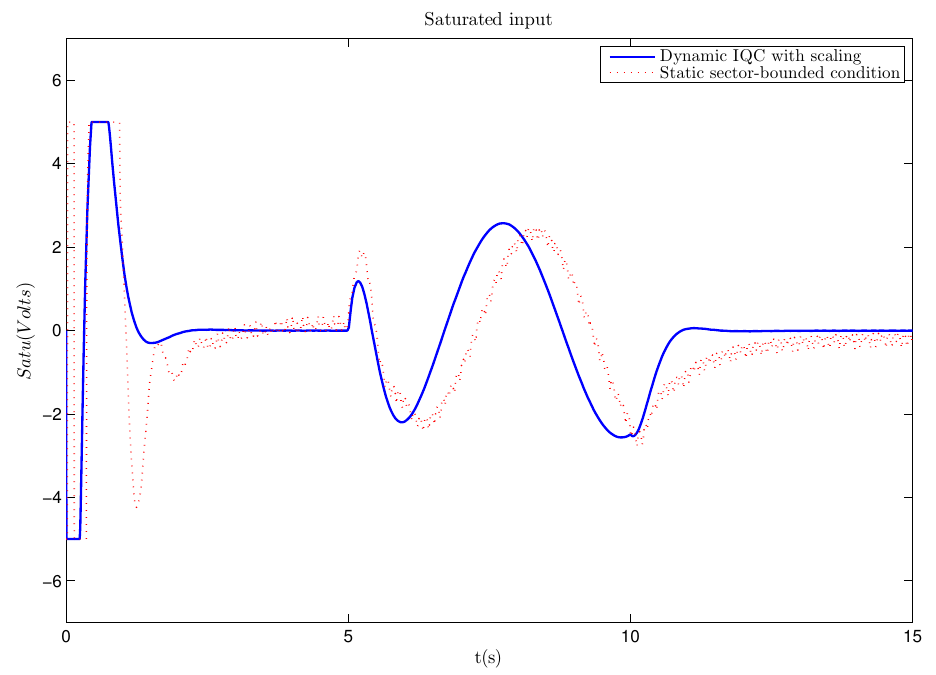}
	%\vspace{-1cm}
	\caption{Saturated voltage of the cart-spring pendulum. Dynamic IQC with scaling factor (bold solid line); static sector-bound condition (thin dotted line). \label{Fig:7}}
\end{figure}

After clustering the poles in the region $\mathrm{Re}(s)<-1$ and imposing an upper bound on the matrix variable $Q$ to avoid numerical singularity, the closed-loop poles are obtained as $-96.5074,~-6.5045\pm17.0188i,-3.0656\pm3.0384i,-1.2863,-1.0000,-1.3579$. In this case, controller synthesis condition (\ref{eq:25}) is solved and a set of dynamic IQC-based state-feedback controller parameters are derived by $F_c=[3658.4~ -2837.8 ~-2404.7~ 1056.0~ -52.5~ -20.6~ -37.2~ -9.1]$ and $H_c = 8495.4$. We assume that the initial conditions are given by $x_0=[0.55m~0~\frac{\pi}{2}rad~0]^{\rm T}$, and the upper bound $\bar{u}$ of actuator output is $5Volts$. Simulation results are shown in Figs. \ref{Fig:6}-\ref{Fig:7}. It can be seen from Fig. \ref{Fig:6} that the pendulum initially undergoes a large swing, during which the cart moves rapidly to recover balance.

\section{Conclusion}

This paper has proposed an $\mathcal{H}_\infty$ state-feedback control method for linear systems subject to time-varying structured uncertainties and input saturation within an IQC-based dissipation framework. By introducing a loop transformation, the original uncertain LFT system is reformulated as a multilayer interconnection involving both uncertainty and nonlinearity blocks. Based on mixed IQC characterizations of the time-varying uncertainties and saturation-induced dead-zone nonlinearity, a new scaled bounded real lemma has been established, from which tractable controller synthesis conditions are derived. In particular, the proposed method achieves improved disturbance attenuation performance in terms of the $\mathcal{L}_2$ gain by solving the resulting synthesis conditions. Two numerical examples, including an uncertain LFT system and a cart--spring--pendulum system, have been presented to validate the effectiveness of the proposed approach. The results show that the proposed controller outperforms both single-IQC-based designs and the static sector-bound method in terms of disturbance attenuation.

Future work will focus on extending the mixed-IQC-based $\mathcal{H}_\infty$ state-feedback design to unstable systems with regional stability requirements. Another direction is to develop alternatives to the current inequality relaxation so as to reduce conservatism and further improve the controller synthesis conditions.

\appendices
\section{$J$-Spectral Factorization}

Following the procedures in \cite{P2015,HP2014}, we use the Popov IQC multiplier as an illustrative example to present the $J$-spectral factorization procedure.

The numerical procedure to construct a $J$-spectral factorization for a para-Hermitian multiplier $\Pi=\Pi^{\sim}$ can be summarized as follows. We first obtain a minimal state-space realization of $\Pi_{\bar{P}}$. The modified Popov IQC \eqref{eq:7} is given by
\begin{align*}
\Pi_{\bar{P}} &= \left[\begin{matrix}
\frac{-0.01\beta^{2}}{(s-\alpha)(s+\alpha)} & \frac{\alpha}{s-\alpha} \\
-\frac{\alpha}{s+\alpha} & 0
\end{matrix}\right] +\left[\begin{matrix}
0 & 1 \\
1 & -0.01
\end{matrix}\right] = \hat{\Pi}_P + D
\end{align*}
where we denote $G(s)=\left[\begin{matrix}
g_1(s) & g_2(s)
\end{matrix}\right]$, with
\begin{align*}
g_1(s) &= \frac{1}{(s-\alpha)(s+\alpha)}\left[\begin{matrix}
-0.01\beta^{2} \\
-\alpha(s-\alpha)
\end{matrix}\right], \nonumber \\
g_2(s) &= \frac{1}{s-\alpha}\left[\begin{matrix}
\alpha \\
0
\end{matrix}\right], \nonumber \\
d_1(s) &= s^{2}-\alpha^{2}, \nonumber \\
d_2(s) &= s-\alpha .
\end{align*}

A controllable canonical-form realization can be written as
\begin{align*}
A_1 &= \left[\begin{matrix}
0 & 1 \\
\alpha^{2} & 0
\end{matrix}\right],\quad
B_{1} = \left[\begin{matrix}
0 \\
1
\end{matrix}\right],\quad
C_1 = \left[\begin{matrix}
-0.01\beta^{2} & 0 \\
\alpha^{2} & -\alpha
\end{matrix}\right], \nonumber \\
A_2 &= \alpha,\quad
B_2 = 1,\quad
C_2= \left[\begin{matrix}
\alpha \\
0
\end{matrix}\right].
\end{align*}
After augmentation, we obtain
\begin{align*}
A &= \left[\begin{matrix}
A_1 & 0 \\
0 & A_2
\end{matrix}\right]
=\left[\begin{array}{cc:c}
0 & 1 & 0 \\
\alpha^{2} & 0 & 0\\
\hdashline
0 & 0 & \alpha
\end{array}\right], \nonumber \\
B &= \left[\begin{matrix}
B_1 & 0 \\
0 & B_2
\end{matrix}\right]
= \left[\begin{array}{cc}
0 & 0 \\
1 & 0 \\
\hdashline
0 & 1
\end{array}\right], \nonumber \\
C &= \left[\begin{matrix}
C_1 & C_2
\end{matrix}\right]
= \left[\begin{array}{cc:c}
-0.01\beta^{2} & 0 & \alpha \\
\alpha^{2} & -\alpha & 0
\end{array}\right].
\end{align*}

To obtain a minimal realization, we examine the observability of $(A,C)$. For $\alpha=1$ in $\frac{1}{s+\alpha}$, the observability matrix is rank-deficient (specifically, $\mathrm{rank}(\mathcal{O})=2<3$). Selecting the first two rows of $\left[\begin{matrix}
C \\
AC
\end{matrix}\right]$ yields
\begin{align*}
S = \left[\begin{matrix}
-0.01 & 0 & 1 \\
1 & -1 & 0
\end{matrix}\right].
\end{align*}
We then construct
\begin{align*}
U = \left[\begin{matrix}
0 & 1 \\
0 & 0 \\
1 & 0.01
\end{matrix}\right],
\end{align*}
which satisfies $SU=I_{2\times 2}$. Therefore, a minimal realization of $\Pi_{\bar{P}}$ is given by
\begin{align*}
A_0 &= SAU = \left[\begin{matrix}
1 & 0.01 \\
0 & -1
\end{matrix}\right], \quad 
B_0 = SB = \left[\begin{matrix}
0 & 1 \\
-1 & 0
\end{matrix}\right], \nonumber \\
C_0 &= CU = \left[\begin{matrix}
1 & 0 \\
0 & 1
\end{matrix}\right], \quad
D_0 = D = \left[\begin{matrix}
0 & 1 \\
1 & -1
\end{matrix}\right].
\end{align*}

Second, we compute a state-space realization $(A_s,B_s,C_s,D_s)$ for the stable part of $\Pi_{\bar{P}}$:
\begin{align*}
A_s &= -1, \quad
B_s = \left[\begin{matrix}
-1 & 0
\end{matrix}\right], \nonumber \\
C_s &= \left[\begin{matrix}
-0.32 \\
1
\end{matrix}\right], \quad
D_s = \left[\begin{matrix}
0 & 1 \\
1 & -0.01
\end{matrix}\right].
\end{align*}
Since $D_s$ is nonsingular, there exists a symmetric solution $X=X^{\rm T}$ to the algebraic Riccati equation (ARE)
\begin{align*}
A_s^{\rm T}X+XA_s-(XB_s+C_s^{\rm T})D_s^{-1}(B_s^{\rm T}X+C_s) = 0.
\end{align*}
For the above data, one obtains
\begin{align*}
X = 0.9950,\qquad L = -0.01,\qquad G = \left[\begin{matrix}
0.99 \\
-1
\end{matrix}\right].
\end{align*}
Let $M$ satisfy $D_s = M^{\rm T}WM$ with $W=\left[\begin{matrix}
1 & 0 \\
0 & -1
\end{matrix}\right]$. One feasible choice is
\begin{align*}
M = \left[\begin{matrix}
1 & 0.4950 \\
1 & -0.5050
\end{matrix}\right].
\end{align*}
Based on $(A_s, B_s, WM^{-{\rm T}}(B_s^{\rm T}X+C_s), M)$, a state-space realization of $\Psi$ is obtained as
\begin{align*}
\Psi_{\bar{P}} = \left[\begin{matrix}
\frac{s+0.5050}{s+1} & 0.4950 \\
\frac{s-0.4950}{s+1} & -0.5050
\end{matrix}\right].
\end{align*}

For the specific form required in the $\mathcal{H}_\infty$ controller synthesis, the general factor $\Psi_{\bar{P}}$ is further converted into $\bar{\Psi}_{\bar{P}}$ (see Section \uppercase\expandafter{\romannumeral3}). The resulting factor is
\begin{align*}
\bar{\Psi}_{\bar{P}} = \left[\begin{matrix}
\frac{1.98s+0.0198}{s+1} & -0.9802 \\
0 & 1
\end{matrix}\right].
\end{align*}

\begin{Rem}\label{Rem:2}
	An ARE is required to compute the $J$-spectral factorization. For the modified Zames--Falb IQC, the matrix $R$ in the ARE may become singular. In that case, a small perturbation $\epsilon$ can be added to the $(1,1)$ entry of $R$ to ensure the existence of a solution.
\end{Rem}

%\begin{IEEEbiography}[{\includegraphics[width=1in,height=1.25in,clip,keepaspectratio]{a1.png}}]{First A. Author} (Fellow, IEEE) and all authors may include 
%biographies. Biographies are
%often not included in conference-related papers.
%This author is an IEEE Fellow. The first paragraph
%may contain a place and/or date of birth (list
%place, then date). Next, the author’s educational
%background is listed. The degrees should be listed
%with type of degree in what field, which institution,
%city, state, and country, and year the degree was
%earned. The author’s major field of study should
%be lower-cased.
%
%The second paragraph uses the pronoun of the person (he or she) and
%not the author’s last name. It lists military and work experience, including
%summer and fellowship jobs. Job titles are capitalized. The current job must
%have a location; previous positions may be listed without one. Information
%concerning previous publications may be included. Try not to list more than
%three books or published articles. The format for listing publishers of a book
%within the biography is: title of book (publisher name, year) similar to a
%reference. Current and previous research interests end the paragraph.
%
%The third paragraph begins with the author’s title and last name (e.g.,
%Dr. Smith, Prof. Jones, Mr. Kajor, Ms. Hunter). List any memberships in
%professional societies other than the IEEE. Finally, list any awards and work
%for IEEE committees and publications. If a photograph is provided, it should
%be of good quality, and professional-looking.
%\end{IEEEbiography}

\begin{IEEEbiography}[{\includegraphics[width=1in,height=1.25in,clip,keepaspectratio]{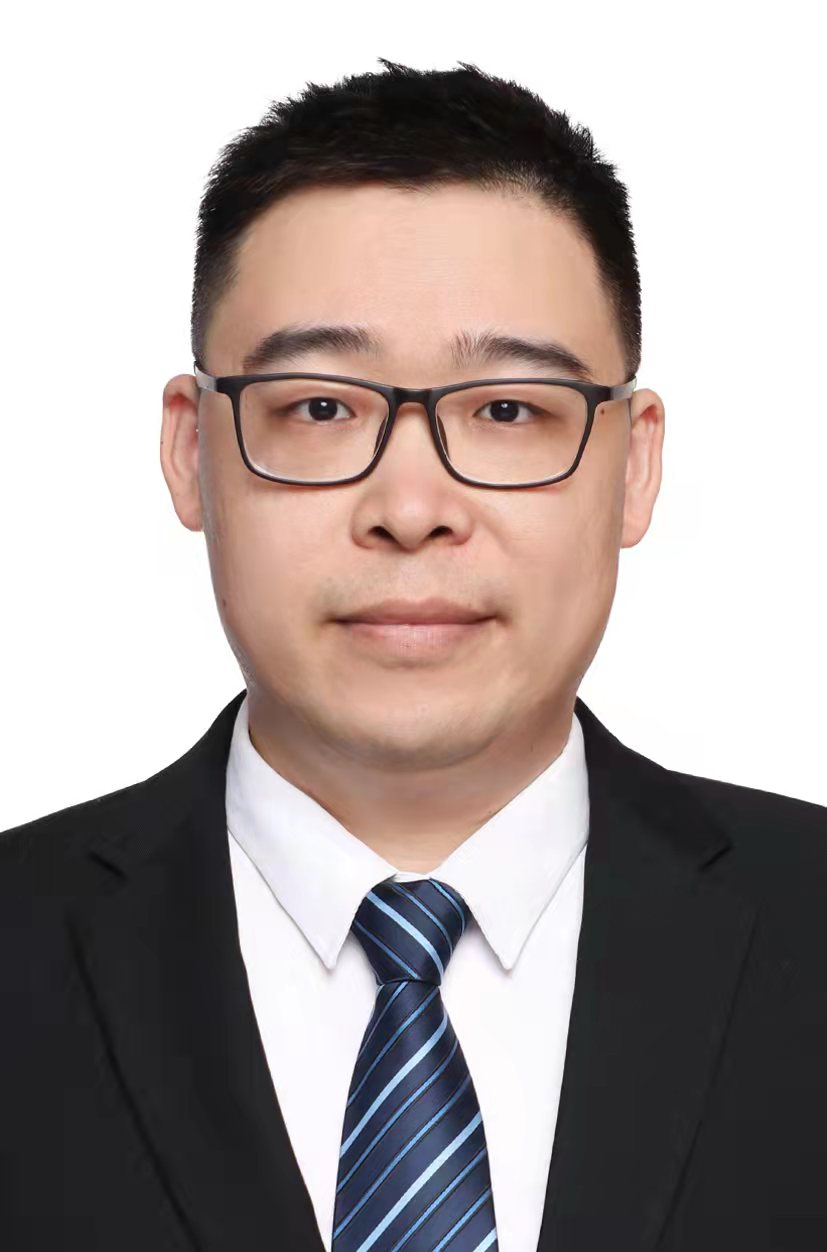}}]{Xu Zhang} is a Ph.D. candidate in the School of Electrical Engineering and Computer Science at the Pennsylvania State University.
	He received the B.S. degree in Automation and the M.S. degree in Control Science and Engineering from Harbin University of Science and Technology in 2008 and Harbin Institute of Technology in 2010, respectively.
	From February 2017 to February 2018, he was a visiting scholar in the School of Mechanical and Aerospace Engineering at North Carolina State University.
	His research interests include dynamic IQC, nonlinear control theory, and the security of cyber-physical systems and machine learning.
\end{IEEEbiography}

\begin{IEEEbiography}[{\includegraphics[width=1in,height=1.25in,clip,keepaspectratio]{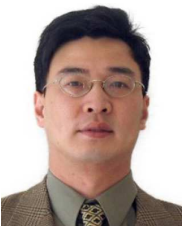}}]{Fen Wu (Senior Member, IEEE)} received the B.S. and M.S. degrees in automatic control
from Beijing University of Aeronautics and Astronautics, Beijing, China, in 1985 and 1988, respectively, and the Ph.D. degree in mechanical engineering from the University of 	California at Berkeley, Berkeley, CA, USA, in 1995. From 1995 to 1997, he was a Research Associate with the Centre for Process Systems Engineering, Imperial College of Science, Technology and Medicine, London, U.K. He was a Staff Engineer with Dynacs Engineering Company Inc., Melbourne, FL, USA, from 1997 to 1999. He is currently a Professor with the Department of Mechanical and Aerospace Engineering, North Carolina State University, Raleigh, NC, USA. His research interests include linear parameter-varying and switching control of nonlinear systems, robust control subject to physical constraints, nonlinear control using sum-of-square programming, distributed
control of multiagent systems, and the application of advanced control and optimization techniques to aerospace, mechanical, and chemical engineering problems. He has authored over 200 conference papers and journal articles. Dr. Wu was an Associate Editor of Journal of Dynamic Systems, Measurement and Control (ASME), IEEE Transactions on Automatic Control, and Automatica. He is currently serving as an Associate Editor for IEEE Transactions on Cybernetics. He is a Fellow of ASME.
	
\end{IEEEbiography}

%\begin{IEEEbiographynophoto}{Second B. Author,} photograph and biography not available at the
%time of publication.
%\begin{IEEEbiographynophoto}{Third C. Author Jr.} (Member, IEEE), photograph and biography not available at the
%time of publication.

\end{document}